\documentclass[orivec,runningheads]{llncs}

\usepackage{amsmath}
\usepackage{hyperref}
\usepackage{graphicx}
\usepackage{color}
\usepackage{amsfonts}
\usepackage{amssymb}
\usepackage[ruled,lined,linesnumbered]{algorithm2e}
\allowdisplaybreaks
\usepackage{tikz}
\SetKw{KwBy}{by}

\newcommand{\Vars}{\ensuremath{\mathcal{V}}}
\newcommand{\vars}[1]{\ensuremath{\Vars(#1)} }

\newcommand{\getcls}[1]{\ensuremath{\Vars_{sym}(#1)}}
\newcommand{\getidx}[1]{\ensuremath{\Vars_{\mathbb{N}}(#1)}}

\newcommand{\Tcal}{\ensuremath{{\cal T}}}

\newcommand{\shiftd}[2]{\ensuremath{\mathit{sh}_{#1}(#2)}}

\newcommand{\irr}[1]{\ensuremath{\mathcal{I}_{#1}} }
\newcommand{\irrV}[1]{\ensuremath{\mathcal{I}^{\Vars}_{#1}}}
\newcommand{\irrSub}[1]{\ensuremath{\sigma^{\mathcal{I}}_{#1}}}

\newcommand{\dom}[1]{\ensuremath{\mathit{dom}(#1)}}
\newcommand{\ran}[1]{\ensuremath{\mathit{ran}(#1)}}
\newcommand{\dep}[1]{\ensuremath{\mathit{dep}(#1)}}
\newcommand{\pos}[1]{\ensuremath{\mathit{pos}(#1)}}
\newcommand{\sub}[1]{\ensuremath{\mathit{sub}(#1)}}
\newcommand{\hd}[1]{\ensuremath{\mathit{head}(#1)}}

\newcommand{\unif}{\ensuremath{\stackrel{{\scriptscriptstyle ?}}{=}}}
\newcommand{\subEQ}[1]{\ensuremath{\mathit{\sigma}_{#1}^{EQ}}}

\newcommand{\storeC}[2]{\ensuremath{\mathit{S}_{#1}(#2)}}
\newcommand{\NstoreC}[2]{\ensuremath{\overline{\mathit{S}_{#1}(#2)}}}

\newcommand{\FuR}[1]{\ensuremath{\mathit{F}^{#1}}}

\newcommand{\stepsubs}[2]{\ensuremath{\Theta^{#1}_{\mathcal{#2}}}}

\newcommand{\useq}[2]{\ensuremath{\mathcal{U}_{#1}^{#2}}}
\newcommand{\uProb}[1]{\ensuremath{\mathcal{U}_{#1}}}
\newcommand{\finCon}[2]{\ensuremath{\mathit{fin}_{#1}(#2)}}

\newcommand{\maxbnd}[2]{\ensuremath{\mathit{maxI}(#1,#2)}}
\newcommand{\minbnd}[2]{\ensuremath{\mathit{minI}(#1,#2)}}
\newcommand{\depbnd}[2]{\ensuremath{\mathit{db}_{#1}(#2)}}

\newcommand{\config}[3]{\ensuremath{\USlangle #1,#2\USrangle_{#3}}}
\newcommand{\USlangle}{\ensuremath{\pmb{\langle}}}
\newcommand{\USrangle}{\ensuremath{\pmb{\rangle}}}

\newcommand{\Rec}[2]{\ensuremath{\mathit{R}_{#1}(#2)}}

\newcommand{\base}[2]{\ensuremath{\mathit{B}_{#1}(#2)}}
\newcommand{\baseSet}[1]{\ensuremath{\mathit{B}_{#1}}}

\newcommand{\depb}{\depbnd{\Theta}{\useq{\Theta}{0}\stepsubs{0}{U}}}
\begin{document}
\title{Schematic Unification}
%
%
\author{David M. Cerna\thanks{This work was supported by the Czech Science Foundation Grant 22-06414L and Cost Action CA20111 EuroProofNet.}\orcidID{0000-0002-6352-603X}}
\authorrunning{D. M. Cerna}
%
\institute{Czech Academy of Sciences Institute of Computer Science \\
\email{dcerna@cs.cas.cz}}
\maketitle              
\begin{abstract}
We present a generalization of first-order unification to a term algebra where variable indexing is part of the object language. We exploit variable indexing by associating some sequences of variables ($X_0,\ X_1,\ X_2,\dots$) with a mapping $\sigma$ whose domain is the variable sequence and whose range consist of terms that may contain variables from the sequence. From a given term $t$, an infinite sequence of terms may be produced by iterative application of $\sigma$. Given a unification problem $U$ and mapping $\sigma$, the \textit{schematic unification problem} asks whether all unification problems $U$, $\sigma(U)$, $\sigma(\sigma(U))$, $\dots$ are unifiable. We provide a terminating and sound algorithm. Our algorithm is \textit{complete} if we further restrict ourselves to so-called $\infty$-stable problems. We conjecture that this additional requirement is unnecessary for completeness. Schematic unification is related to methods of inductive proof transformation by resolution and inductive reasoning. 

\end{abstract}
\section{Introduction}
Consider a sequence of clause sets of the form: $C_n=\{ x\leq x,\  max(x,y)\not \leq z \vee x\leq z,\
 max(x,y)\not \leq z \vee y\leq z\}\cup 
\{ f(x)\not = i \vee \ f(y)\not = i\vee s(x)\not \leq y\mid 0\geq i\geq n\}\cup 
\{ f(x)= 0\vee \dots \vee f(x)=n\}$ where $n\geq 0$ and $x\leq y \equiv x=y \vee x<y$\footnote{Extracted from a sequent calculus proof of the infinitary pigeonhole principle. The full refutation $C_n$ is described in Chapter 6 of David Cerna's Dissertation \href{https://doi.org/10.34726/hss.2015.25063}{doi.org/10.34726/hss.2015.25063}.}. A refutation of any member of this sequence is (not so easily) handled by modern theorem provers; however, refutating $C_n$ as a whole and finitely representing the result is beyond their scope. An important milestone towards developing automated reasoning methods for such clause sets is the development of unification methods for \emph{schema of terms}, a problem we tackle in this paper. 

The origin of such clause sets begins with \emph{cut-elimination}~\cite{Gentzen34} within the classical sequent calculus. Cut-elimination provides the foundation of automated reasoning through \emph{Herbrand's Theorem} and provides a method for \emph{formal proof analysis}~\cite{girard87}. Inversely, methods of cut-elimination based on resolution (CERES) ~\cite{DBLP:journals/jsc/BaazL00} benefit from the state-of-the-art developments in the field of automated reasoning, providing an efficient approach to cut-elimination, and hinting at the possibility of a framework for inductive proof analysis~\cite{DBLP:journals/tcs/BaazHLRS08}. While CERES has been extended, allowing for simple inductive proof analysis together with an effective extension of \emph{Herbrand's Theorem}~\cite{DBLP:journals/logcom/LeitschPW17}, the limiting factor is the state-of-the-art in inductive theorem proving~\cite{DBLP:journals/fuin/AravantinosEP13,DBLP:journals/jar/BensaidP14,DBLP:books/daglib/0013358}. Inductive proof analysis results in clause sets similar to $C_n$ and requires a form of \emph{schematic unification} to refute.

The first formal description of schematic unification was presented in~\cite{DBLP:conf/cade/CernaL16,DBLP:journals/jar/CernaLL21}, and a special case was handled in~\cite{DBLP:conf/synasc/Cerna21}. While the proof-theoretic origins of this problem are far removed from practice, one can imagine applying this type of unification within areas concerned with stream reasoning and inductive/co-inductive types and structures. We present an algorithm for the unifiability of uniform schematic unification problems (Definition~\ref{def.classmapping}) that is both sound and terminating and is complete when the unification problem is restricted to being  $\infty$-stable (Definition~\ref{def:infstable}). We conjecture that every uniform schematic unification problem is  $\infty$-stable. 

The paper is as follows:  in Section~\ref{sec:prelims}, we introduce the basic notation; in Section~\ref{sec:schematicUnif}, we introduce the schematic unification problem, in Section~\ref{sec:findepunif}, we introduce $\Theta$-unification, in Section~\ref{sec:propTheta}, we discuss properties of $\Theta$-unification, in Section~\ref{sec:decide}, we
present our algorithm for the unifiability of uniform schematic unification, and in Section~\ref{sec:conclud}, we discuss open problems and future work.

\subsection{Related work}\label{subsec:relatedwork}
 One notable example of a related investigation is the unification of \emph{Primal Grammars}~\cite{HERMANN1997111}. However, unlike the problem investigated in this work, unification of primal grammars only considers a finite number of existential (unifiable) variables, which may occur infinitely often in the schematization. Instead, we consider an infinite number of existential variables, each of which may occur only a finite number of times. This issue extends to \emph{P-grammars}~\cite{AravantinosCP10}, an extension of primal grammars. 

The \textit{Schematic Substitution} (Definition~\ref{def.classmapping}) is essentially a set of rewrite rules, and thus, is related to \emph{E-unification} and \emph{narrowing}~\cite{ESCOBAR2009103}. However, unlike narrowing methods, \emph{(i)} our term rewriting system is non-terminating (but confluent), \emph{(ii)} we apply rewrite rules to variables, not terms, \emph{(iii)} rule applications may introduce variables which are not fresh (so-called {\em extra variables}) implying that we would need to consider narrowing over a conditional rewrite system~\cite{DBLP:journals/scp/CholewaEM15}. As with primal grammars, restricting the variables occurring on the right-hand side of rewrite rules is essential to the decidability of its unification problem.

Our setting is similar to first-order cyclic \emph{term graphs}~\cite{DBLP:journals/mscs/BaldanBCK07,HABEL1995110}. For example, unification of cyclic term graphs in the nominal setting is presented in~\cite{DBLP:journals/fuin/Schmidt-Schauss22}. However, we handle variables differently from these earlier investigations. Anti-unification over first-order cyclic term graphs~\cite{DBLP:conf/rta/BaumgartnerKLV18} has also been studied. However, anti-unification typically introduces variables rather than solving for the appropriate substitution~\cite{DBLP:journals/corr/abs-2302-00277}. Thus, the infinite depth of the terms involved plays a less significant role than it does when unifying such structures.  
Similarly, there have been a few investigations into unification of \textit{rational trees}~\cite{DBLP:journals/constraints/ColmerauerD03,DBLP:conf/fsttcs/RamachandranH93}, i.g. infinite trees with a finite number of subtrees. Our problem considers a more general class of infinite trees: (i) may contain infinitely many variables, each occurring a finite number of times, (ii) containing infinitely many subtrees but only finitely many of them differ modulo variable renaming.  

Another related unification problem where decidability is contingent on strict conditions on variable occurrence is \emph{cycle unification}~\cite{DBLP:conf/cade/BibelHW92}, unification between self-referencing \emph{Horn clauses}. This unification problem is undecidable and thus differs from our work as we provide a decision procedure.

\section{Preliminaries}
\label{sec:prelims}
Let $\Vars_{sym}$ be a countably infinite set of variable symbols, variables are elements of $\Vars= \{ X_i\mid X\in \Vars_{sym}\wedge i\in\mathbb{N}\}$, $\Sigma$  a first-order signature such that each symbol $f\in \Sigma$ is associated with an arity $n\geq 0$ (symbols with arity $0$ are referred to as constants). By $\Tcal$, we denote the term algebra constructed from $\Vars$ and $\Sigma$. The head of a term $t\in \Tcal$, denoted $\hd{t}$, is defined as $\hd{f(t_1,\cdots, t_n)} =f$ and $\hd{x} = x$ for $x\in \Vars$.  When dealing with sequences of terms $t_1,\cdots, t_n$, we abbreviate this sequence as $\overline{t_n}$.
We will denote variable symbols using uppercase letters,  $X, Y, Z, X^1, Y^2, Z^3$, variables using lowercase letters when the index is not relevant, $x, y, z, x^1, y^2, z^3$, and using uppercase letters when the index is relevant $X_i, Y_i, Z_i, X^1_i, Y^2_i, Z^3_i$ where $i\in\mathbb{N}$. Variable symbols are always denoted using uppercase letters unless otherwise specified. For $X_i\in \Vars$, $\getcls{X_i}= \{X\}$ and  $\getidx{X_i}= \{i\}$. We extend these functions to  $t\in \Tcal$ as follows: $\getcls{f(\overline{t_m})}= \bigcup_{i=1}^{m} \getcls{t_i}$,\ $\getidx{f(\overline{t_m})}= \bigcup_{i=1}^{m} \getidx{t_i}$. The set of variables occurring in a term $t$ is denoted by $\vars{t}$. 

The set of \emph{positions} of a term $t\in \Tcal$, 
denoted by $\pos{t}$, is the set of strings of 
finitely many positive integers, defined as 
$\pos{x}=\{\epsilon \}$ and 
$\pos{f(\overline{t_n})}=\{\epsilon \}\cup  
\bigcup_{i=1}^n \{i. p \mid  p\in \pos{t_i}\wedge 
1\leq i \leq n\}$, where $\epsilon$ denotes the 
empty string, $f\in \Sigma$ has arity $n$, and 
$x\in \Vars$. For example, the term at position 
$\epsilon$ of $f(x,a)$ is itself, and at position 
$2$ is $a$. 



Let $p\in \pos{t}$, then by $t\vert_p$ we denote the term at position $p$ of $t$. 
We define the \textit{depth of a term $t$} recursively as follows:  $\dep{x}= 1$ where $x\in \Vars$, and $\dep{f(t_1,\cdots, t_n)}= 1 + \max\{\dep{t_i}\mid 1\leq i \leq n\}$.   We can use the set of positions of a term $t$ to construct the subterms of $t$, $\sub{t} = \{ t\vert_p\mid p\in \pos{t}\}$. 


A \emph{substitution} is a mapping from $\Vars$ to $\Tcal$ such that all but finitely many variables are mapped to themselves. Lowercase Greek letters are used to denote them unless otherwise stated. Substitutions are extended to terms in the usual way using postfix notation, writing $t\sigma$ for an \emph{instance} of a term $t$ under a substitution $\sigma$. The \emph{composition} of substitutions $\sigma$ and $\vartheta$, written as juxtaposition $\sigma\vartheta$, is the substitution defined as $x(\sigma\vartheta) = (x\sigma)\vartheta$ for all variables $x$. The \textit{domain} of a substitution $\sigma$,  denoted $\dom{\sigma}$, is defined as the set of variables $\{x\mid x\sigma \not =x\}$. The \textit{range} of a substitution $\sigma$,  denoted $\ran{\sigma}$, is defined as  $\{ x\sigma\mid x\in \dom{\sigma}\}$.

We assume familiarity with the basics of unification theory, see~\cite{BaaderS01}. For example, $s\unif t$, where $s,t\in\Tcal$, denotes a \textit{unification equations}. We refer to $x\unif t$ as a \textit{binding} if $x\in \Vars$, i.e. $x$ binds to $t$ (also written $x\mapsto t$). A  unification problem is a set of \textit{unification equations}. Given a unification problem $U$,  $\mathit{unif}(U)$ denotes the \textit{most general unifier (mgu)} of $U$ (if it exists)  computed by a first-order unification algorithm applied to $U$ or $\bot$ if $U$ is not unifiable (i.g.  Martelli-Montanari algorithm~\cite{MartelliM82}, or Robinson~\cite{Robinson1965}).

\section{Schematic Unification}
\label{sec:schematicUnif}
In this section, we present the construction of schematic unification problems through the variable indexing mechanism introduced in the previous section.
\begin{definition}[Index Shift]\label{def.subsshift}
Let $t\in \Tcal$ and $d\geq 0$. Then $\shiftd{d}{t}$ is defined inductively as follows: $\shiftd{d}{X_j}=X_{j+d}$ and $\shiftd{d}{f(\overline{t_m})}=f(\shiftd{d}{t_1},\cdots,\shiftd{d}{t_m)}$.
\end{definition} 
\begin{example}
    Let $t = f(X_1,g(Y_3,X_4))$. Then $\shiftd{0}{t}=f(X_1,g(Y_3,X_4))$, $\shiftd{2}{t}=f(X_3,g(Y_5,X_6))$, and $\shiftd{5}{t}=f(X_6,g(Y_8,X_9))$
\end{example}

\begin{definition}[Substitution Schema]\label{def.classmapping}
Let  $\overline{X^n}\in \Vars_{sym}$ be pairwise distinct and $\overline{t_n} \in \Tcal$. Then $\Theta = \bigcup_{i=1}^{n} \{ X_j^i\mapsto  \shiftd{j}{t_i}\mid j\geq 0\}$ is a \emph{substitution schema} with $\dom{\Theta}=\{\overline{X^n}\}$ and $\base{\Theta}{X^i}=t_i$ for $1\leq i\leq n$. Furthermore, $\baseSet{\Theta} = \bigcup_{X\in \dom{\Theta}}\base{\Theta}{X}$ is the \emph{base} of $\Theta$, and for $X\in \dom{\Theta}$,
\begin{itemize}
    \item  $R_{\Theta}(X)= \{i\mid x\in\vars{\base{\Theta}{X}}\ \& \ \getcls{x} = \{X\}\ \& \ \getidx{x} = \{i\}\}$.
    \end{itemize}
We define the following categories of substitution schema:
\begin{itemize}
    \item $\Theta$ is \emph{simple} if 
for every $t\in \baseSet{\Theta}$, $|\getcls{t}\cap\dom{\Theta}|\leq 1$
\item $\Theta$ is \emph{uniform} if \emph{(i)} it's simple, and \emph{(ii)} for $X\in \dom{\Theta}$,
$\vert R_{\Theta}(X)\vert \leq 1$.
\item $\Theta$ is  \emph{primitive} if \emph{(i)} it's uniform and \emph{(ii)} for $X\in \dom{\Theta}$,
$R_{\Theta}(X) \subset\{0,1\}$.
\end{itemize}
\end{definition} 
Given a variable $x\in \Vars$ and schematic substitution $\Theta$, we write $x\in \Theta$ ($x\not \in \Theta$) when $\getcls{x} = \{X\}$ and $X\in \dom{\Theta}$ ($X\not \in \dom{\Theta}$).
\begin{example}
Consider the following Substitution Schema (\textbf{bold} highlights why the substitution schema is not simple, uniform, or primitive):
\begin{itemize}
    \item $\Theta_1 = \{L_i \mapsto h(h(X_i,h(\textbf{R}_{i+4},X_{i})),\textbf{L}_{i+1})\mid i\geq 0\}\cup  \{R_i \mapsto h(h(h(R_{i+2},X_{i}),$ $R_{i+1})\mid i\geq 0\}$ is \textbf{not} \emph{simple}.
\item $\Theta_2 = \{L_i \mapsto h(h(X_i,h(\textbf{L}_{i+4},X_{i})),\textbf{L}_{i+1})\mid i\geq 0\}\cup  \{R_i \mapsto h(h(\textbf{R}_{i+2},X_{i}),$ $\textbf{R}_{i+1})\mid i\geq 0\}$ is \emph{simple} but \textbf{not} \emph{uniform}. 
\item $\Theta_3 = \{L_i \mapsto h(h(X_i,h(\textbf{L}_{i+4},X_{i})),\textbf{L}_{i+4})\mid i\geq 0\}\cup  \{R_i \mapsto h(h(\textbf{R}_{i+2},X_{i}),$ $\textbf{R}_{i+2})\mid i\geq 0\}$ is \emph{uniform} but \textbf{not} \emph{primitive}. 
\item $\Theta_4 = \{L_i \mapsto h(h(X_i,h(L_{i+1},X_{i})),L_{i+1})\mid i\geq 0\}\cup  \{R_i \mapsto h(h(R_{i},X_{i}),$ $R_{i})\mid i\geq 0\}$ is \emph{primitive}, i.e. $R_{\Theta}(L) =\{1\}$ and $R_{\Theta}(R)=\{0\}$. 
\end{itemize}

\end{example}

While \textit{uniform} substitution schemata seem more expressive than \textit{primitive}, the following Theorem shows that  \textit{uniform} substitution schemata can be transformed into primitive substitution schemata. 
\begin{theorem}\label{thm:primitiveEnough}
Let $\Theta$ be a uniform substitution schema such that $\vert \dom{\Theta} \vert = n$. Then there exists substitutions $\sigma_1,\cdots, \sigma_n$ and  primitive substitution schema $\Xi$ such that \emph{(i)} $\dom{\Xi} = \dom{\Theta}$ and \emph{(ii)} for all $X\in \dom{\Xi}$, $\base{\Xi}{X} = \base{\Theta}{X}\sigma_1\cdots \sigma_n$.
\end{theorem}
\begin{proof}
To simplify the argumentation below let $$\mathit{NP}_{\Theta} = \{X\mid X\in \dom{\Theta} \ \& \ R_{\Theta}(X)= \{j\} \ \& \ j>1\}.$$ If  $\mathit{NP}_{\Theta}=\emptyset$, then $\Theta$ is already primitive and $\sigma_1,\cdots \sigma_n$ are identity substitutions. Let us assume, for the induction hypothesis, that when $|\mathit{NP}_{\Theta}| =m$ for  $0\leq m<n$, there exists substitutions $\sigma_1,\cdots \sigma_n$ such that $\Theta'$ is primitive. We show that this statement holds for $|\mathit{NP}_{\Theta}| =m+1$. 

Let $X\in \mathit{NP}_{\Theta}$ such that $R_{\Theta}(X) =\{m\}$ and for all $Y\in \mathit{NP}_{\Theta}$, $R_{\Theta}(Y) =\{j\}$ and $m\geq j$. Now for each $0 \leq k < m$ and $Y\in \getcls{\base{\Theta}{X}}$ such that $Y\not = X$, we define 
$$g_m(k,Y) = \{Y_{i+j}\mid j= m\cdot l+(k\mod m) \ \&\ l\geq 0 \ \&\ Y_{i+j}\in \vars{\baseSet{\Theta}}\}.$$ 
For each $g_m(k,Y)$ we associate an $X^{(k,Y)}\in \Vars_{sym}$ such that (i) $X^{(k,Y)} \not \in \getcls{\baseSet{\Theta}}$ and (ii) $X^{(k,Y)} = X^{(k',Y')}$ iff $k = k'$ and $Y=Y'$. 



Now we construct a substitution $\sigma$ as follows: (i) $X_{i+m}\sigma = X_{i+1}$, (ii) for every $1 \leq k< m$, $Z\in\getcls{\base{\Theta}{X}}$ such that $Z\not = X$, and $Z_{i+w}\in g_m(k,Z)$ such that  $ w = m\cdot l+(k\mod m)$, $(Z_{i+w})\sigma = X^{(k,Z)}_{i+l}$, and (iii) $\sigma$ maps all other variables to themselves. 

Now let $\Theta'$  be a substitution schema such that $\dom{\Theta'} = \dom{\Theta}$ and \emph{(ii)} for all $X\in \dom{\Theta'}$, $\base{\Theta'}{X} = \base{\Theta}{X}\sigma$.

Note that $|\mathit{NP}_{\Theta'}| =m$ and thus by the induction hypothesis there exists substitutions $\sigma_1,\cdots, \sigma_n$ and primitive substitution schema  $\Xi$ such that \emph{(i)} $\dom{\Xi} = \dom{\Theta}$ and \emph{(ii)} for all $X\in \dom{\Xi}$, $\base{\Xi}{X} = \base{\Theta}{X}\sigma_1\cdots \sigma_n$.
\end{proof}

\begin{example}
    Consider the following schematic substitution $\Theta =$

    $$\begin{array}{l} \left\lbrace\begin{array}{l|c}
          L_i \Leftarrow f(f(X_{i},f(Z_{i+1},f(E_{i+3},f(X_{i+1},f(W_{i+11},X_{i+1}))))),L_{i+1})\ & \ i\geq 0\end{array}\right\rbrace \cup \\
          \left\lbrace \begin{array}{l|c} S_i \Leftarrow  f(f(E_{i},f(W_{i+3},E_{i+2})),S_{i+4})\ & \ i\geq 0\end{array}\right\rbrace \cup \\
         \left\lbrace \begin{array}{l|c} R_i \Leftarrow  f(f(W_{i},W_{i+4}),R_{i+7})\ & \ i\geq 0\end{array}\right\rbrace 
    \end{array}$$

We start by converting $R_i \Leftarrow  f(f(W_{i},W_{i+4}),R_{i+7})$. 
Observe that $g_7(0,W) = \{W_i\}$, $g_7(3,W) = \{W_{i+3}\}$,  $g_7(4,W) = \{W_{i+4},W_{i+11}\}$, and the rest are empty. Thus, we derive the following substitution 
$$\sigma_1 = \left\lbrace\begin{array}{l}
  R_{i+7}\mapsto R_{i+1}, \ W_i\mapsto X^{(0,W)}_{i},  \ W_{i+3}\mapsto X^{(3,W)}_{i}, \ W_{i+4}\mapsto X^{(4,W)}_{i},  \\
  W_{i+11}\mapsto X^{(4,W)}_{i+1}
\end{array}\right\rbrace$$

Applying $\sigma_1$ to $\Theta$ results in $\Theta' =$
 $$ \begin{array}{l} \left\lbrace\begin{array}{l|c}
          L_i \Leftarrow f(f(X_{i},f(Z_{i+1},f(E_{i+3},f(X_{i+1},f(X^{(4,W)}_{i+1},X_{i+1}))))),L_{i+1})\ & \ i\geq 0\end{array}\right\rbrace \cup \\
         \left\lbrace \begin{array}{l|c} S_i \Leftarrow  f(f(E_{i},f(X^{(3,W)}_{i},E_{i+2})),S_{i+4})\ & \ i\geq 0\end{array}\right\rbrace \cup \\
        \left\lbrace \begin{array}{l|c}  R_i \Leftarrow  f(f(X^{(0,W)}_{i},X^{(4,W)}_{i}),R_{i+1})\ & \ i\geq 0\end{array}\right\rbrace 
    \end{array}$$
Now we convert $S_i \Leftarrow  f(f(E_{i},f(X^{(3,W)}_{i},E_{i+2})),S_{i+4})$. Observe that $g_4(0,E) = \{E_i\}$, $g_4(3,E) = \{E_{i+3}\}$,  $g_4(2,W) = \{E_{i+2}\}$,  $g_4(0,X^{(3,W)}) = \{X^{(3,W)}_{i}\}$, and the rest are empty. Thus, we derive the following substitution 
$$\sigma_2 = \left\lbrace\begin{array}{l}
  S_{i+4}\mapsto S_{i+1}, \ E_i\mapsto X^{(0,E)}_{i},  \ E_{i+3}\mapsto X^{(3,E)}_{i}, \ E_{i+2}\mapsto X^{(2,E)}_{i},  \\
  X^{(3,W)}_i\mapsto X^{(0,X^{(3,W)})}_{i}
\end{array}\right\rbrace$$

Applying $\sigma_2$ to $\Theta$ results in $\Theta'' =$
 $$ \begin{array}{l}
        \left\lbrace \begin{array}{l|c}  L_i \Leftarrow f(f(X_{i},f(Z_{i+1},f( X^{(3,E)}_{i},f(X_{i+1},f(X^{(4,W)}_{i+1},X_{i+1}))))),L_{i+1})\ & \ i\geq 0\end{array}\right\rbrace \cup \\
         \left\lbrace \begin{array}{l|c} S_i \Leftarrow  f(f(X^{(4,E)}_{i},f( X^{(4,X^{(3,W)})}_{i},X^{(2,E)}_{i})),S_{i+1})\ & \ i\geq 0\end{array}\right\rbrace \cup \\
        \left\lbrace \begin{array}{l|c}  R_i \Leftarrow  f(f(X^{(7,W)}_{i},X^{(4,W)}_{i}),R_{i+1})\ & \ i\geq 0\end{array}\right\rbrace \end{array}$$
Observe that $\Theta''$ is a primitive schematic substitution. 
\end{example}
For the rest of the paper, when considering a uniform substitution schema, we assume it is primitive. We build substitutions from substitution schema by constructing a restriction based on the variables in a given term. 
\begin{definition}[$t$-substitution]\label{def.tsub}
 Let $t\in \Tcal$ and $\Theta$ a substitution schema. Then $\Theta^t$ is a substitution satisfying: for all $x\in \Vars$, if  $x\in \vars{t}$ and  $\{x\mapsto s\} \in \Theta$, then $x\Theta^t = s $, otherwise $x\Theta^t =x$.
\end{definition}

\begin{example}\label{ex:topstep2}
Consider the term $t= f(f(L_2,X_0),L_1)$ and the substitution schema  $\Theta =\{L_i\mapsto f(f(X_{i},X_{i}),L_{i+1})\mid i\geq 0\}$. Note, $\Theta$ is primitive as $R_{\Theta}(L)= \{1\}$. The $t$-substitution of $\Theta$ is as follows:  $\Theta^t= \{L_1\mapsto f(f(X_1,X_1),L_{2}),L_2$ $\mapsto f(f(X_2,X_2),L_{3})\}$.

\end{example}

\begin{definition}[$\Theta$-instances]\label{def.sequence}
Let $i\geq 0$, $t\in \Tcal$, and $\Theta$ be a substitution schema. Then the \emph{i$^{th}$ $\Theta$-instance of $t$}, denoted $\Theta^t(i)$, is defined inductively as follows: $\Theta^t(0)=t$ and  $\Theta^t(i+1)=s\Theta^s$ where $s= \Theta^t(i)$.
\end{definition}

\begin{example}
Continuing with Example~\ref{ex:topstep2}, Some $\Theta$-instances of $f(f(L_2,X_0),L_1)$ are as follows: $\Theta^t(0)=f(f(L_2,X_0),L_1)$, $\Theta^t(1)= f(f(L_2\Theta^{L_2},X_0), L_1\Theta^{L_1}) = f(f(f(f(X_2,X_2),L_3),X_0), f(f(X_1,X_1),L_2))$. Continuing the process, $\Theta^t(2)=$  $f(f(f($ $f(X_2,X_2),f(f(X_3,X_3),L_4)),X_0), f(f(X_1,X_1),f(f(X_2,X_2),L_3)))$.
\end{example}

\begin{definition}[Schematic Unification Problem]\label{def.lup}
Let $\mathcal{U}$ be a set of unification equations and $\Theta$  substitution schema. Then the pair $(\mathcal{U},\Theta)$ form a \emph{schematic unification problem}, denoted $\uProb{\Theta}$. We refer to $\uProb{\Theta}$ as \emph{simple}, \emph{uniform}, or \emph{primitive} if $\Theta$ is \emph{simple}, \emph{uniform}, or \emph{primitive}.


\end{definition} 
\begin{example}
     Continuing with Example~\ref{ex:topstep2}, $\uProb{\Theta}=\{L_0 \unif f(Y_0,Y_0), f(f(L_2,X_0),$ $ L_1)\unif f(L_0,Y_1), f(f(L_2,X_0),L_1)\unif f(L_0,h(Y_0))\}$ is primitive as $\Theta$ is primitive.
\end{example}

\begin{definition}[Schematically Unifiable]\label{def.lupunif}
Let $\uProb{\Theta}$ be a schematic unification problem. Then $\mathcal{U}$ is \emph{schematically unifiable} iff for all $i\geq 0$, $\useq{\Theta}{i}=\{\Theta^t(i) \unif \Theta^s(i)\mid t\unif s\in \mathcal{U}\}$ is unifiable. 
\end{definition} 
Note that the unifier of a schematic unification problem may have an infinite domain. Thus, it must be described as a substitution schema. We focus on developing a procedure for unifiability of \emph{ uniform schematic unification problems} and leave a finite representation of the unifier to future work. However, our algorithm (Algorithm~\ref{alg:procedure}) returns the objects required to construct a unifier.

\begin{definition}[Schematic Unifier]
Let $\uProb{\Theta}$ be a schematic unification problem, $\Xi$ a substitution schema, and $\sigma$ a substitution. Then  $(\Xi,\sigma)$ is a \emph{schematic unifier} of $\uProb{\Theta}$ if for all $t \unif s\in \uProb{\Theta}$, $i\geq 0$,  $v\in \{s,t\}$, there exists $m\geq i$ such that   for all $k\geq m$, $\Theta^v(i)\mu$ subsumes $\Xi^{v\sigma}(k)$, where $\mu_{i}$ is the mgu of $\Theta^t(i) \unif \Theta^s(i)$ and $\mu = \{ x\mapsto t\mid x \not \in \Theta \ \& \ x\mu_i=t\}$. 
\end{definition}

    \begin{example}\label{ex:lupunif}
       Continuing with the substitution schema of example~\ref{ex:topstep2}, consider the schematic unification problem $\{L_0 \unif f(Y_0,Y_0)\}$. Observe that for all $i\geq 0$, $\Theta^s(i) = s$ where $s=f(Y_0,Y_0)$ and all the instance problems are unifiable, i.e. $ L_0 \unif f(Y_0,Y_0)$, $ f(f(X_0,X_0),L_1) \unif f(Y_0,Y_0)$, etc., and thus by Definition~\ref{def.lupunif} $L_0 \unif f(Y_0,Y_0)$ is schematically unifiable. A  primitive schematic unifier is $(\{L_i\mapsto f(L_{i+1},L_{i+1})\mid i\geq 0\},\{Y_0\mapsto L_1,\ L_0\mapsto f(L_1,L_1)\})$. 
    \end{example}

    \begin{example}\label{ex:lupunif2}
        For a more complex example, consider the schematic unification problem $L_0 \unif f(Y_0,f(Y_1,Y_0))$ where $\Theta = \{ L_i\mapsto f(f(X_{i},f(X_{i+1},X_{i})),L_{i+1})\mid i\geq 0\}$. This problem is also schematically unifiable and a schematic unifier is $(\{L_i\mapsto f(L_{i+1},f(L_{i+2},L_{i+1}))\mid i\geq 0\},\{Y_0\mapsto L_1,\ Y_1\mapsto L_2,\ L_0 \mapsto f(L_{1},f(L_{2},$ $L_{1}))\})$. Observe that the schematic unifier is simple but not primitive. 
    \end{example}

\section{$\Theta$-Unification}
\label{sec:findepunif}
We now introduce $\Theta$-Unification, which computes the transitive-symmetric closure of unification problems decompositionally derivable from a set of unification problems $\mathcal{U}$. Observe that $\Theta$-Unification keeps, for each variable $x$, all terms $t$ such that $x\unif t$ occurs during the unification process; this is necessary as some variables may occur in both the current instance of a schematic unification problem as well as future instances and, thus, influence the unifier. Our goal is to isolate a set of bindings which captures the recursive structure of the sequence of unification problems.  

In essence, our algorithm  constructs a set of bindings $\mathcal{S}$ such that for all $\{x\unif t\} \in \mathcal{S}$, one of the following holds: (i) $x\in 
\Theta$ (\textbf{Store}, \textit{condition 1}), (ii) there exists $\{y\unif s\}\in \mathcal{S}$ such that $y,t\in \Vars$ and $t=y$ or $x=y$ (\textbf{Store}, \textit{condition 2}), (iii) there exists $\{y\unif s\}\in \mathcal{S}$ such that $x\in \vars{s}$ (\textbf{Store}, \textit{condition 2}), or (iv) there exists $z\in \vars{\mathcal{S}}$ and $k\geq 0$ such that $z\in \Theta$ and $x\in \Theta^{z}(k)$ (\textbf{Store}, \textit{condition 3}). A \textit{configuration} is a pair of sets of the form $\config{\mathcal{S}}{\mathcal{U}}{\Theta}$ where $\mathcal{S}$ and $\mathcal{U}$ contain unification equations. Given a set $\mathcal{U}$, the  \textit{initial configuration} is of the form  $\config{\emptyset}{\mathcal{U}}{\Theta}$. The rules are presented in Table~\ref{tab:infrules}. 
\begin{table}
    \centering
\begin{itemize}
\item \textbf{Decomposition}
        $\config{\mathcal{S}}{\mathcal{U} \cup \{f(\overline{r_n})\unif f(\overline{s_n})\}}{\Theta}   \Rightarrow
\config{\mathcal{S}}{\mathcal{U}\cup \{  r_i\unif s_i \mid  1\leq i \leq n\}}{\Theta}$

\item \textbf{Orient-1} $\config{\mathcal{S}}{\mathcal{U} \cup \{ r \unif x\}}{\Theta}   \Rightarrow \config{\mathcal{S}}{\mathcal{U} \cup \{x \unif r\}}{\Theta}$ where $x\in \Vars$, $r\not \in \Vars$.
\item \textbf{Orient-2} $\config{\mathcal{S}}{\mathcal{U} \cup \{ y \unif x \}}{\Theta}   \Rightarrow \config{\mathcal{S}}{\mathcal{U} \cup \{y\unif x,x \unif y\}}{\Theta}$ where  $x\unif y \not \in \mathcal{U}$ and $x,y\in \Vars$.

\item \textbf{Transitive} $\config{\mathcal{S}}{\mathcal{U} \cup \{ x\unif r, x\unif s\}}{\Theta}    \Rightarrow \config{\mathcal{S}}{\mathcal{U} \cup \{ x \unif r, x\unif s,r\unif s\}}{\Theta}  $ where  $r\not = s$, $x\in \Vars$

\item \textbf{Reflexive} $\config{\mathcal{S}}{\mathcal{U} \cup \{ t\unif t\}}{\Theta}    \Rightarrow \config{\mathcal{S}}{\mathcal{U}}{\Theta}  $  
\item \textbf{Clash} $\config{\mathcal{S}}{\mathcal{U}\cup \{ y\unif r\}}{\Theta} \Rightarrow \bot$ where $r,s\not \in \Vars$ and $\hd{r} \not = \hd{s}$.

\item \textbf{Store} $\config{\mathcal{S}}{\mathcal{U}\cup \{ y\unif r\}}{\Theta}    \Rightarrow \config{\mathcal{S}\cup \{ y\unif r\}}{\mathcal{U}}{\Theta}  $ where $y\in \Vars$, and one of the following conditions holds:
    \begin{itemize}
        \item[1)] $y \in \Theta$,  and  if $r\in \Vars$, then $r \in \Theta$.
        \item[2)] $y \not \in \Theta$, and there exists $\{x\unif s'\} \in  \mathcal{S}$ such that $y\in \vars{s'}$, $r=x$, or $y=x$.
        \item[3)] $ y \not \in \Theta$, and  there exists $z\in \vars{\mathcal{S}}$ and $k\geq 0$ s.t. $z \in \Theta$ and $y\in \vars{\Theta^{z}(k)}$.
\end{itemize}

\end{itemize}
    \caption{Inference Rules of $\Theta$-unification}
    \label{tab:infrules}
\end{table}
\begin{table}
    \centering
 \begin{itemize}
\item[$\bullet$] \textit{State}($\mathcal{U},\Theta$) -- Creates a state object and a initial configuration $\config{\emptyset}{\mathcal{U}}{\Theta}$.
\item[$\bullet$] \textsf{state}.\textit{orient1}($\config{\mathcal{S}}{\mathcal{U}}{\Theta}$,  $s\unif x$) -- The \textsf{state} object guarentees that $s\unif x\in \mathcal{U}$. Applies \textbf{Orient-1} to the unification equation $s\unif x$, returns the resulting configuration and adds  $x\unif s$ to \textsf{state}.\textsf{changes}. 
\item[$\bullet$] \textsf{state}.\textit{orient2}($\config{\mathcal{S}}{\mathcal{U}}{\Theta}$,  $y\unif x$) -- The \textsf{state} object guarentees that $y\unif x\in \mathcal{U}$ and $x\unif y\not \in \mathcal{U}$. Applies \textbf{Orient-2} to the unification equation $y\unif x$, returns the resulting configuration and adds  $x\unif y$ to \textsf{state}.\textsf{changes}. 
\item[$\bullet$] \textsf{state}.\textit{decomposition}($\config{\mathcal{S}}{\mathcal{U}}{\Theta}$,  $s\unif t$) -- The \textsf{state} object guarentees that $s\unif t\in \mathcal{U}$. Applies \textbf{Decomposition} to the unification equation $s\unif t$ resulting in the set of unification equations $\mathcal{U}'$. Applies \textbf{Reflexitive} $\mathcal{U}'$ when applicable resulting in $\mathcal{U}''$. Finally, returns the resulting configuration, extended by $\mathcal{U}''$,  and adds $\mathcal{U}''$ to \textsf{state}.\textsf{changes}.
\item[$\bullet$] \textsf{state}.\textit{transitivity}($\config{\mathcal{S}}{\mathcal{U}}{\Theta}$,  $x\unif s$,  $x\unif t$) -- The \textsf{state} object guarentees that $x\unif s,x\unif t\in \mathcal{U}$ and $s\unif t\not\in \mathcal{U}$. Applies \textbf{Transitivity} to the unification equations $x\unif s$ and $x\unif t$, applies \textbf{Reflexitive} if $s = t$, returns the resulting configuration and add $s\unif t$ to   \textsf{state}.\textsf{changes} if $s\not = t$.
\item[$\bullet$] \textit{filter}(\textit{func}, \textsf{state}.\textsf{changes}) -- \textit{func} is a unary function which takes a unification equation as an argument and returns \textbf{True} or \textbf{False}. Returns the subset of \textsf{state}.\textsf{changes}  for which \textit{func} returns \textbf{True}.
\item[$\bullet$]  \textsf{state}.\textit{orient1Cond}($s\unif t$) -- checks if \textbf{Orient-1} may be applied to $s\unif t$
\item[$\bullet$]  \textsf{state}.\textit{orient2Cond}($s\unif t$) -- checks if \textbf{Orient-2} may be applied to $s\unif t$
\item[$\bullet$]   \textsf{state}.\textit{decomCond}($s\unif t$) -- checks if \textbf{Decomposition} may be applied to $s\unif t$
\item[$\bullet$]   \textsf{state}.\textit{leftVarNewEQCond}($s\unif t$) -- checks if $s\in \Vars$ and $s\unif t\not \in \textsf{state}.\textsf{varDict}[s]$.
\item[$\bullet$]   \textsf{state}.\textit{storeCond}($s\unif t$) -- checks if \textbf{Store} may be applied to $s\unif t$ 
\item[$\bullet$] \textsf{state}.\textit{store}($\config{\mathcal{S}}{\mathcal{U}}{\Theta}$,  $s\unif t$) -- The \textsf{state} object guarentees that $s\unif t\in \mathcal{U}$ and $s\unif t\not\in \mathcal{S}$. Applies \textbf{Store} to $s\unif t$ and returns the resulting configuration.
 \end{itemize}
    \caption{Sub-procedures used by Algorithm~\ref{alg:Thetaprocedure}\&~\ref{alg:updateprocedure}}
    \label{tab:subprocedures}
\end{table}
Observe that arbitrary application of these inference rules may lead to non-termination as \textbf{Transitivity} can introduce unification equations that were previously decomposed. To avoid non-termination, $\Theta$-unification is performed using the procedure outlined in Algorithm~\ref{alg:Thetaprocedure} \& Algorithm~\ref{alg:updateprocedure}. However, we will avoid using these Algorithm~\ref{alg:Thetaprocedure} \& Algorithm~\ref{alg:updateprocedure} within examples and for proofs found in the Section~\ref{sec:propTheta} as the procedure requires extensive book-keeping and obfuscates the derivation of a final configuration. 

\begin{algorithm}  
\SetAlgoLined\DontPrintSemicolon
\SetKwProg{myalg}{Procedure}{}{}
 \myalg{$\mathit{thUnif}(\mathcal{U}',\Theta)$\mbox{ \textbf{:}}}{ 
\textsf{state}, $\config{\mathcal{S}}{\mathcal{U}}{\Theta}$ $\leftarrow$ \textit{State}($\mathcal{U}',\Theta$) \tcp*{New state object and configuration}
 \While{\textit{update}\emph{(\textsf{state})}}{
        \While{$s\unif x$ $\leftarrow$ \emph{\textsf{state}.\textsf{orient1}.\textit{pop}()}}{
           $\config{\mathcal{S}}{\mathcal{U}}{\Theta}$ $\leftarrow$ \textsf{state}.\textit{orient1}($\config{\mathcal{S}}{\mathcal{U}}{\Theta}$,  $s\unif x$)
        }
        \While{$f(\overline{r_n})\unif f(\overline{s_n})$ $\leftarrow$ \emph{\textsf{state}.\textsf{decom}.\textit{pop}()}}{
           $\config{\mathcal{S}}{\mathcal{U}}{\Theta}$ $\leftarrow$ \textsf{state}.\textit{decomposition}($\config{\mathcal{S}}{\mathcal{U}}{\Theta}$, $f(\overline{r_n})\unif f(\overline{s_n})$)
        }
        \While{$x\unif y$ $\leftarrow$ \emph{\textsf{state}.\textsf{orient2}.\textit{pop}()}}{
           $\config{\mathcal{S}}{\mathcal{U}}{\Theta}$ $\leftarrow$ \textsf{state}.\textit{orient2}($\config{\mathcal{S}}{\mathcal{U}}{\Theta}$, $y\unif x$)
        }
        \While{$x\unif s,x\unif t$ $\leftarrow$ \emph{\textsf{state}.\textsf{trans}.\textit{pop}()}}{
         $\config{\mathcal{S}}{\mathcal{U}}{\Theta}$ $\leftarrow$ \textsf{state}.\textit{transitivity}($\config{\mathcal{S}}{\mathcal{U}}{\Theta}$,  $x\unif s$, $x\unif t$)
        }
 }
\While{\emph{\textsf{state}.\textsf{Store}} $\leftarrow$ \emph{\textit{filter}(\emph{\textsf{state}.}\textit{storeCond}, $\mathcal{U}$)} $\not = \emptyset$}{
\While{$x\unif t$ $\leftarrow$ \emph{\textsf{state}.\textsf{Store}.\textit{pop}()}}{
              $\config{\mathcal{S}}{\mathcal{U}}{\Theta}$ $\leftarrow$ \textsf{state}.\textit{store}($\config{\mathcal{S}}{\mathcal{U}}{\Theta}$, $x\unif t$)
}
}
\textit{unif}($\mathcal{S}\cup \mathcal{U}$)\tcp*{Throws Exception Clash or Cycle}
\Return  $\config{\mathcal{S}}{\mathcal{U}}{\Theta}$
}
  \caption{Unification Algorithm for $\Theta$-Unification}
\label{alg:Thetaprocedure}
\end{algorithm}

\begin{algorithm}  
\SetAlgoLined\DontPrintSemicolon
\SetKwProg{myalg}{Procedure}{}{}
 \myalg{$\mathit{update}$\emph{(\textsf{state})}\mbox{\textbf{:}}}{ 
  \If{\textit{clash}\emph{(\textsf{state}.\textsf{changes})}}{
  \textbf{Raise} \textsf{clashException}
 }
\textsf{state}.\textsf{orient1} $\leftarrow$ \textit{filter}(\textsf{state}.\textit{orient1Cond}, \textsf{state}.\textsf{changes})\\
\textsf{state}.\textsf{decom} $\leftarrow$ \textit{filter}(\textsf{state}.\textit{decomCond}, \textsf{state}.\textsf{changes})\\
\textsf{state}.\textsf{orient2} $\leftarrow$ \textit{filter}(\textsf{state}.\textit{orient2Cond}, \textsf{state}.\textsf{changes})\\
\For{$x\unif t \in$ \emph{\textit{filter}(\textsf{state}.\textit{leftVarNewEQCond}, \textsf{state}.\textsf{changes})} }{
\textsf{state}.\textsf{trans} $\leftarrow$ [ ($x\unif t$, $eq'$) \textbf{for} $eq'$ \textbf{in}  \textsf{state}.\textsf{varDict}[x] ]\\
\textsf{state}.\textsf{varDict}[x] $\leftarrow$ \textsf{state}.\textsf{varDict}[x] $\cup$ \{$x\unif t$\}
}

\textsf{state}.\textsf{changes} $\leftarrow$ $\emptyset$\\
\If{\emph{\textsf{state}.\textsf{orient1} $\cup$ \textsf{state}.\textsf{decom} $\cup$ \textsf{state}.\textsf{orient2} $\cup$ \textsf{state}.\textsf{trans}} $= \emptyset$}{
\Return \textsf{False}
}
\Return  \textsf{True}
}
  \caption{Update Algorithm for $\Theta$-Unification}
\label{alg:updateprocedure}
\end{algorithm}

By $\config{\mathcal{S}}{\mathcal{U}}{\Theta}\Rightarrow^* \config{\mathcal{S}'}{\mathcal{U}'}{\Theta}$, we denote the derivation of the configuration $\config{\mathcal{S}'}{\mathcal{U}'}{\Theta}$ from $\config{\mathcal{S}}{\mathcal{U}}{\Theta}$ applying the inference rules of Table~\ref{tab:infrules} in accordance with the procedure outlined in Algorithm~\ref{alg:Thetaprocedure}. By $\finCon{\Theta}{\mathcal{U}}$, we denote the \textit{final configuration} returned by Algorithm~\ref{alg:Thetaprocedure} when applied to the initial configuration $\config{\emptyset}{\mathcal{U}}{\Theta}$. The final \textit{store} is denoted by $\storeC{\Theta}{\mathcal{U}}$. If Algorithm~\ref{alg:Thetaprocedure} throws an exception, then $\finCon{\Theta}{\mathcal{U}}=\bot$ ($\storeC{\Theta}{\mathcal{U}}=\bot$). 
\begin{theorem}[Termination]\label{thm:terminationTheta}
    Let $\mathcal{U}$ be a finite set of unification problems and $\Theta$ a schematic substitution. Then, $\mathit{thUnif}(\mathcal{U},\Theta)$ (Algorithm~\ref{alg:Thetaprocedure}) terminates after finitely many steps. 
\end{theorem}
\begin{proof}
If an exception is thrown (Line 22 of Algorithm~\ref{alg:Thetaprocedure} and Line 3 of Algorithm~\ref{alg:updateprocedure}), the procedure terminates.
Observe that the inference rules of Table~\ref{tab:infrules} do not introduce fresh variables. Thus, the number of variables remains constant during the application of Algorithm~\ref{alg:Thetaprocedure} to the set of unification equations $\mathcal{U}$. Also, the set of all subterms of terms contained in $\mathcal{U}$ is finite. Thus, for any variable $x\in \vars{\mathcal{U}}$, there is a finite set of unification equations $x\unif t$ which may be introduced during the application of Algorithm~\ref{alg:Thetaprocedure} to $\mathcal{U}$. Thus, at some point, $state.varDict[x]$ will contain all possible unification equations and no pairs of the form $(x\unif t,x\unif s)$
will be added to \textsf{state}.\textsf{trans}.

Similarly, \textbf{Decomposition} can only be applied finitely many terms as the maximum term depth is finite. Furthermore, \textbf{Orient-1} and \textbf{Orient-2} can only be applied finitely many times for the same reason as \textbf{Transitivity}; this implies that after finitely many steps, Algorithm~\ref{alg:updateprocedure} will return \textbf{False}. Observe that the loop on Lines 17-21 of Algorithm~\ref{alg:Thetaprocedure} will terminate after finitely many steps as the number of unification equations in $\mathcal{U}$ is finite and \textbf{Store} removes equations from $\mathcal{U}$.
\end{proof}

\begin{theorem}[Soundness]
\label{Thm:soundnessTheta}
    Let $\mathcal{U}$ be a unifiable set of finite unification problems. Then $\storeC{\Theta}{\mathcal{U}} \not =\bot$.
\end{theorem}
\begin{proof}
 Observe that all transformations applied to $\mathcal{U}$ by Algorithm~\ref{alg:Thetaprocedure} are valid equational transformations. 
\end{proof}
\begin{theorem}[Completeness]
    Let $\mathcal{U}$ be a finite set of unification problems. Then if   $\storeC{\Theta}{\mathcal{U}} \not =\bot$, then $\mathcal{U}$ is unifiable. 
\end{theorem}
\begin{proof}
    If $\mathit{unif}(\mathcal{U})\not = \bot$, then no  \textbf{Clash} was discovered and \textit{occurs check} did not fail. Observe that Algorithm~\ref{alg:Thetaprocedure} computes the transitive-symmetric closure modulo decomposition and will find a \textbf{Clash} if one exists in $\mathcal{U}$.
\end{proof}
An essential observation for the results presented below is as follows:
\begin{proposition}
    Let $\mathcal{U}$ be a set of  unification equations such that $\finCon{\Theta}{\mathcal{U}} =\config{\mathcal{S}'}{\mathcal{U}'}{\Theta}$. Then for all $t\unif s\in \mathcal{U}'\cup \mathcal{S}'$, $t\in \Vars$. 
\end{proposition}
\begin{proof}
Observe that \textbf{Store} adds unification equations to the store if the left side is a variable. Thus, $\mathcal{S}'$ will only contain unification equations of the required form. If $\mathcal{U}'$ contained a unification equation $t\unif s$ such that $t\not \in \Vars$, then either we can apply (i)  \textbf{Orient 1} , (ii)  \textbf{Decomposition}, (iii) \textbf{Reflexivity}, (iv) or \textbf{Clash}. In any of these cases, the application of the rule contradicts the assumption that  $\config{\mathcal{S}'}{\mathcal{U}'}{\Theta}$ is a final configuration.
\end{proof}
\begin{example}\label{ex:oneTheta}
    Consider the schematic unification problem $\mathcal{U}=\{L_0 \unif h(Y_0,h(Y_1,$ $Y_0))\}$ where $\Theta = \{L_i \mapsto h(h(X_i,h(X_{i+1},X_i)),L_{i+1})\mid i\geq 0\}$.  Computation of $\storeC{\Theta}{\mathcal{U}_1}$ (See Definition~\ref{def.lupunif}) proceeds as follows:  $\mathcal{U}_1=\{h(h(X_0,h(X_{1},X_0)),L_{1}) \unif h(Y_0,h($ $Y_1,Y_0))\}$ and 
    $$ \config{\emptyset}{\mathcal{U}_1}{\Theta}\Rightarrow^{\textbf{Decomposition}} \config{\emptyset}{\mathcal{U}^1}{\Theta}$$ Where  $\mathcal{U}^1 = \{ h(X_0,h($ $X_1,X_0))$ $ \unif Y_0,\ L_1 \unif h(Y_1,Y_0)\}$. Next we can apply \textbf{Orient 1}
 to $h(X_0,h(X_1,X_0$ $)) \unif Y_0$ resulting in 
  $$ \config{\emptyset}{\mathcal{U}^1}{\Theta}\Rightarrow^{\textbf{Orient 1}} \config{\emptyset}{\mathcal{U}^2}{\Theta}$$
 where  $\mathcal{U}^2= \{ Y_0 \unif h(X_0,h(X_1,X_0)),\ L_1 \unif h(Y_1,$ $Y_0)\}$.
At this point we apply \textbf{Store} to the contents of $\mathcal{U}^2$ 
  $$ \config{\emptyset}{\mathcal{U}^2}{\Theta}\Rightarrow^{\textbf{Store}} \config{\mathcal{S}^1}{\{ Y_0 \unif h(X_0,h(X_1,X_0))\}}{\Theta}\Rightarrow^{\textbf{Store}}\config{\mathcal{S}^2}{\emptyset}{\Theta}$$
Where $\mathcal{S}^1= \{ L_1 \unif h(Y_1,Y_0)\}$, $\mathcal{S}^2 = \{ Y_0 \unif h(X_0,h(X_1,X_0)),\ L_1 \unif h(Y_1,Y_0)\}$. No rules are applicable to $\config{\mathcal{S}^2}{\emptyset}{\Theta}$ and thus $\finCon{\Theta}{\mathcal{U}_1} = \config{\mathcal{S}^2}{\emptyset}{\Theta}$ and $\storeC{\Theta}{\mathcal{U}_1} = \{ Y_0 \unif h(X_0,h(X_1,X_0)),\ L_1 \unif h(Y_1,Y_0)\}$. 
\end{example}
\begin{example}
\label{ex:theta2}
Continuing with Example~\ref{ex:oneTheta}, we compute $\storeC{\Theta}{\storeC{\Theta}{\mathcal{U}_1}\Theta^{L_1}(1)}$ as follows:  
$\storeC{\Theta}{\mathcal{U}_1}\Theta^{L_1}(1) = \{ Y_0 \unif h(X_0,h(X_1,X_0)),\ h(h(X_1,h(X_{2},X_1)),L_{2}) \unif h(Y_1$ $,Y_0)\}$ and
   $$ \config{\emptyset}{\storeC{\Theta}{\mathcal{U}_1}\Theta^{L_1}(1)}{\Theta}\Rightarrow^{\textbf{Decomposition}} \config{\emptyset}{\mathcal{U}^1}{\Theta}$$
 where  $\mathcal{U}^1 = \{ Y_0 \unif h(X_0,h($ $X_1,X_0)),$ $ L_2 \unif Y_0,\ h(X_1,h(X_2,X_1)) \unif Y_1\}$. Next we can apply \textbf{Orient 2}
 to $L_2 \unif Y_0$ resulting in 
  $$ \config{\emptyset}{\mathcal{U}^1}{\Theta}\Rightarrow^{\textbf{Orient 2}} \config{\emptyset}{\mathcal{U}^2}{\Theta}$$
where $\mathcal{U}^2=  \{ Y_0 \unif h(X_0,h(X_1,X_0)),$ $ h(X_1,h(X_2,X_1)) \unif Y_1,\ L_2 \unif Y_0,\ Y_0 \unif L_2\}$.
Next we can apply \textbf{Orient 1}
 to $ h(X_1,h(X_2,X_1)) \unif Y_1$ resulting in 
$$ \config{\emptyset}{\mathcal{U}^2}{\Theta}\Rightarrow^{\textbf{Orient 1}} \config{\emptyset}{\mathcal{U}^2}{\Theta}$$
where  $\mathcal{U}^2=  \{Y_0 \unif h(X_0,h(X_1,X_0)), Y_1 \unif h(X_1,h(X_2,X_1)) $ $,\ L_2 \unif Y_0,\ Y_0 \unif L_2\}$, $r_1 = h(X_1,h(X_2,X_1))$. Next we apply \textbf{Transitivity} between $Y_0 \unif L_2$ and $Y_0 \unif h(X_0,h(X_1,X_0))$ resulting in 
$$ \config{\emptyset}{\mathcal{U}^2}{\Theta}\Rightarrow^{\textbf{Transitive}} \config{\emptyset}{\mathcal{U}^3}{\Theta}$$
where  $\mathcal{U}^3=  \{L_2 \unif h(X_0,h(X_1,X_0)),\ Y_0 \unif h(X_0,h(X_1,X_0)), L_2 \unif Y_0,\ $ $ Y_1 \unif h(X_1,h(X_2,X_1)),\ Y_0 \unif L_2\}$.
At this point we apply \textbf{Store} three times resulting in $\storeC{\Theta}{\storeC{\Theta}{\mathcal{U}_1}\Theta^{L_1}(1)} = \{ Y_0 \unif L_2,\ L_2 \unif h(X_0,h(X_1,X_0)),\ $ $Y_0 \unif h(X_0,h(X_1,X_0))\}$.
\end{example}
Using $\Theta$-unification to compute $\storeC{\Theta}{\mathcal{U}_2}$ we get 
$\{L_2 \unif h(X_1,h(X_2,X_1)),Y_0\unif L_2, \ Y_0 \unif h(X_1,h(X_2,X_1))\}$ which is equivalent to  $\storeC{\Theta}{\storeC{\Theta}{\mathcal{U}_1}\Theta^{L_1}(1)}$. Furthermore, $\storeC{\Theta}{\mathcal{U}_3} = \{ X_1 \unif h(X_2,h(X_3,X_2)),\ L_3 \unif h(X_1,X_0)\}$. Observe that $\storeC{\Theta}{\mathcal{U}_3}\sigma = \storeC{\Theta}{\mathcal{U}_1}$ where  $\sigma=\{ X_4\mapsto X_2 , L_3\mapsto L_1 , X_2\mapsto Y_1 , X_3\mapsto X_1 , X_1\mapsto Y_0\}$. Essentially, we map $\storeC{\Theta}{\mathcal{U}_3}$ to $\storeC{\Theta}{\mathcal{U}_1}$ implying that if $\storeC{\Theta}{\mathcal{U}_1}$ is unifiable, then $\storeC{\Theta}{\mathcal{U}_3}$  is unifiable. At this point, we can deduce that for all $i\geq 3$, $\storeC{\Theta}{\mathcal{U}_i}$ is unifiable. This is key to  Algorithm~\ref{alg:procedure} and Theorem~\ref{Thm:algTerm}\ \& \ \ref{thm:algsoundness}.

\section{Properties of $\Theta$-Unification}
\label{sec:propTheta}
We now present important properties of $\Theta$-Unification essential to the termination and correctness of Algorithm~\ref{alg:procedure}.

\begin{definition}[Irrelevant Set]\label{def:irrelevantSet}
     Let $\uProb{\Theta}$ be a schematic unification problem, $i\geq 0$ such that $\storeC{\Theta}{\useq{\Theta}{i}}\not = \bot$, and $\finCon{\Theta}{\useq{\Theta}{i}} =\config{\storeC{\Theta}{\useq{\Theta}{i}}}{\mathcal{U}}{\Theta}$. Then we define \emph{the i$^{th}$ irrelevant set} as $\irr{i} = \{r_1\unif r_2 \mid r_1\unif r_2\in \mathcal{U}\ \& \ r_1\not\in\Theta\}.$ Furthermore, let $\irrSub{i}= \mathit{unif}(\irr{i})$, that is the unifier derived from $\irr{i}$.
\end{definition}

Observe that the definition of irrelevant set avoids equations of the form $L_2\unif X_5$; such equations are not members of $\storeC{\Theta}{\useq{\Theta}{i}}$ (\textbf{Store}, \textit{condition 2}).
\begin{example}
Consider the derivation presented in Example~\ref{ex:theta2}. The irrelevant set computed from $\storeC{\Theta}{\storeC{\Theta}{\mathcal{U}_1}\Theta^{L_1}(1)}$ is  $\irr{2} = \{Y_1 \unif h(X_1,h(X_2,X_1))\}$. We prove below  that $\storeC{\Theta}{\storeC{\Theta}{\mathcal{U}_1}\Theta^{L_1}(1)} = \storeC{\Theta}{\useq{\Theta}{2}}$ (Lemma~\ref{lem:delayedSubs}).
\end{example}
\begin{definition}[Step Substitution]\label{def:stepSub}
Let $\uProb{\Theta}$ be a schematic unification problem and $i\geq 0$ such that $\storeC{\Theta}{\useq{\Theta}{i}}\not = \bot$. Then the \emph{step substitution} is defined as
$\stepsubs{i}{U}=\bigcup_{t\in C}\Theta^{t}(1)$ where $C = \{x\mid x\in \Theta \ \& \ x\in \vars{\storeC{\Theta}{\useq{\Theta}{i}}}\}$.
\end{definition}
\begin{example} 
Consider the schematic unification problem  $\mathcal{U}$
        $$f(X_4,L_0) \unif f(f(f(A_{0},f(R_3,B_{0})),R_0),f(S_0,f(C_0,D_0)))$$
  with schematic substitution $\Theta=$ 
    $$\begin{array}{l}
          \left\lbrace \begin{array}{l|c} L_i \Leftarrow f(f(X_{i},f(Z_{i+1},f( E_{i},f(X_{i+1},f(B_{i+1},X_{i+1}))))),L_{i+1}) \ & \ i\geq 0\end{array}\right\rbrace \cup \\
          \left\lbrace \begin{array}{l|c} S_i \Leftarrow  f(f(F_{i},f( G_{i},H_{i})),S_{i+1})\ & \ i\geq 0\end{array}\right\rbrace\cup \\
         \left\lbrace \begin{array}{l|c}  R_i \Leftarrow  f(f(A_{i},B_{i}),R_{i+1})\ & \ i\geq 0\end{array}\right\rbrace
    \end{array}$$
Observe that $\storeC{\Theta}{\useq{\Theta}{0}}= $
$$\{ X[4]\unif f(f(A_{0},f(R_3,B_{0})),R_0),\ L_0\unif f(S_0,f(C_0,D_0))\}$$
The step substitution derived from $\storeC{\Theta}{\useq{\Theta}{0}}$ is 

$$\stepsubs{0}{U} = \left\lbrace \begin{array}{l}
          L_0 \Leftarrow f(f(X_{0},f(Z_{1},f( E_{0},f(X_{1},f(B_{1},X_{1}))))),L_{1})\\
          S_0 \Leftarrow  f(f(F_{0},f( G_{0},H_{0})),S_{1})\\
          R_0 \Leftarrow  f(f(A_{0},B_{0}),R_{1})\\
          R_3 \Leftarrow  f(f(A_{3},B_{3}),R_{4})\\
    \end{array}\right\rbrace$$

\end{example}
Using the above definitions, we show that equations in the irrelevant sets, modulo substitution, remain irrelevant for large instances. 

\begin{lemma}\label{lem:cumulativeI}
     Let $\uProb{\Theta}$ be a schematic unification problem, $i\geq 0$ s.t. $\storeC{\Theta}{\useq{\Theta}{i}}\not = \bot$ and $\storeC{\Theta}{\useq{\Theta}{i+1}}\not = \bot$. Then $\irr{i}\stepsubs{i}{U}\subseteq \irr{i+1}$.
\end{lemma}
\begin{proof}
   Let $\finCon{\Theta}{\useq{\Theta}{i}} =\config{\storeC{\Theta}{\useq{\Theta}{i}}}{\mathcal{U}}{\Theta}$. Observe, for every $r\in \ran{\stepsubs{i}{U}}, r\irrSub{i} $ $= r$ as otherwise for some unification equations in $\irr{i}$ we could have applied \textbf{Store}, \textit{condition 3}. Furthermore, $\useq{\Theta}{i+1} = \useq{\Theta}{i}\stepsubs{i}{U}$. Using $\Theta$-unification, we construct the  derivation $\config{\emptyset}{\useq{\Theta}{i+1}}{\Theta}\Rightarrow^* \config{\emptyset}{\mathcal{U}\stepsubs{i}{U}}{\Theta}\Rightarrow^* \config{\storeC{\Theta}{\useq{\Theta}{i}}}{\mathcal{U}'}{\Theta} \Rightarrow^* \config{\storeC{\Theta}{\useq{\Theta}{i+1}}}{\mathcal{U}''}{\Theta}$. Note that $\irr{i}\stepsubs{i}{U}\subseteq \mathcal{U}'$ and thus, is contained in $\mathcal{U}''$ as applying $\stepsubs{i}{U}$ to $\irr{i}$ does not change the left side of bindings contained in $\irr{i}$, i.e. \textbf{Store} does not apply.  Thus, $\irr{i}\stepsubs{i}{U}\subseteq \irr{i+1}$.
\end{proof}

The recursive construction presented in the following Lemma captures what was discussed in Example~\ref{ex:theta2} at the end of Section~\ref{sec:findepunif}.
\begin{lemma}\label{lem:delayedSubs}
    Let $\uProb{\Theta}$ be a schematic unification problem and $i\geq 0$ such that $\storeC{\Theta}{\useq{\Theta}{i}}\not = \bot$. Then $
    \storeC{\Theta}{\storeC{\Theta}{\useq{\Theta}{i}}\stepsubs{i}{U}}=  \storeC{\Theta}{\useq{\Theta}{i+1}}$.
\end{lemma}
\begin{proof}
    Let $\finCon{\Theta}{\useq{\Theta}{i}}=\config{\storeC{\Theta}{\useq{\Theta}{i}}}{\mathcal{U}}{\Theta}$. From the 
    construction of  $\useq{\Theta}{i}$ and $\useq{\Theta}{i+1}$ there exists a sequence of rules such that  $\config{\emptyset}{\useq{\Theta}{i+1}}{\Theta}\Rightarrow^* \config{\emptyset}{\mathcal{U}\stepsubs{i}{U}}{\Theta}$ and 
    $\mathcal{U}\stepsubs{i}{U}= \irr{i}\stepsubs{i}{U}\cup \storeC{\Theta}{\useq{\Theta}{i}}\stepsubs{i}{U} \cup N\stepsubs{i}{U} $ where 
    $N=\{r_1\unif r_2\mid r_1,r_2\in \Vars\ \&\ r_1\in \Theta \ \& \ r_2\not \in 
    \Theta\}$. We know $\irr{i}\stepsubs{i}{U}\subseteq \irr{i+1}$ and thus it can be ignored. Concerning $N$, it contains 
    unification equations which \textbf{Store}, \textit{condition 1} avoids adding to the store. Thus, after applying $\stepsubs{i}{U}$, 
    $x\stepsubs{i}{U}\unif y$ would be removed in favor of $y\unif x\stepsubs{i}{U}$ which may already be in $\storeC{\Theta}{\useq{\Theta}
    {i}}\stepsubs{i}{U}$. Thus, $N\stepsubs{i}{U}$ can be ignored as it will only result in unification equations already in  
    $\storeC{\Theta}{\useq{\Theta}{i}}\stepsubs{i}{U}$. Furthermore,  $\config{\emptyset}{\storeC{\Theta}{\useq{\Theta}
    {i}}\stepsubs{i}{U}}{\Theta}\Rightarrow^* \config{\storeC{\Theta}{\useq{\Theta}{i+1}}}{\mathcal{U}'}{\Theta}$ Thus, $\storeC{\Theta}{\storeC{\Theta}{\useq{\Theta}{i}}\stepsubs{i}{U}} = \storeC{\Theta}{\useq{\Theta}{i+1}}$. 
\end{proof}
There are two important corollaries of Lemma~\ref{lem:delayedSubs}: 

\begin{corollary}\label{cor:alwaysempty}
    Let $\uProb{\Theta}$ be a schematic unification problem and $i\geq 0$ such that $\storeC{\Theta}{\useq{\Theta}{i}}= \emptyset$. Then for all $j\geq i$, $\storeC{\Theta}{\useq{\Theta}{j}}=  \emptyset$.
\end{corollary}

\begin{corollary}\label{cor:alwaysfail}
    Let $\uProb{\Theta}$ be a schematic unification problem and $i\geq 0$ such that $\storeC{\Theta}{\useq{\Theta}{i}}= \bot$. Then for all $j\geq i$, $\storeC{\Theta}{\useq{\Theta}{j}}=  \bot$.
\end{corollary} 
\subsection{Bounding Term Depth}

The recursive application of $\Theta$-unification to a schematic unification problem $\uProb{\Theta}$ greatly simplifies the computation of $\storeC{\Theta}{\useq{\Theta}{i}}$ for $i\geq 0$; we only need to consider bindings which are relevant to larger instances. In addition, we can bound the depth of terms that occur in a \textit{normalimalized} form of $\storeC{\Theta}{\useq{\Theta}{i}}$.

\begin{definition}[$\Theta$-normalization]
 Let $\Theta$ be a schematic substitution and $r\in\Tcal$. We define the \emph{$\Theta$-normalization} of the term $r$, abbreviated $r\uparrow_\Theta$,  inductively as follows: \emph{(i)} if $r\in \Vars$ then  $r\uparrow_{\Theta} =r$, \emph{(ii)} if there exists $d,k> 0$ and $L\in \dom{\Theta}$ such that $r= \Theta^{L_d}(k)$ , then $r\uparrow_\Theta = L_d$, \emph{(iii)} otherwise, if $r=f(r_1,\dots,r_n)$ for $f\in \Sigma$, then $r\uparrow_t =f(r_1\uparrow_\Theta,\dots,r_n\uparrow_\Theta)$.
\end{definition}
\begin{example}
    Consider  $\Theta = \{L_i \mapsto h(h(X_i,h(X_{i+1},L_{i+1})),X_i)\}$ and the term $r=h(X_0,h(X_{1},h(h(X_1,h(X_{2},L_{2})),X_1)))$. Then $r\uparrow_{\Theta} =h(X_0,h(X_{1},L_1))$  as $r\vert_{2.2} = h(h(X_1,h(X_{2},L_{2})),X_1) = \Theta^{L_1}(1)$.
\end{example}

\begin{definition}[Normalized Store]\label{def:normalized}
    Let $\uProb{\Theta}$ be a schematic unification problem and $i\geq 0$. Then $\NstoreC{\Theta}{\useq{\Theta}{i}} =\{x\unif t\uparrow_\Theta \mid  x\unif t\in \storeC{\Theta}{\useq{\Theta}{i}}\}$.
\end{definition}
By considering $\NstoreC{\Theta}{\useq{\Theta}{i}}$ rather than $\storeC{\Theta}{\useq{\Theta}{i}}$ we can bound the size of terms 
occurring in $\NstoreC{\Theta}{\useq{\Theta}{i}}$. To simplify the following discussion, consider a set of unification equations 
$\mathcal{U}$. Now, let $\mathit{terms}(\mathcal{U})= \{ r_i\mid i\in \{1,2\} \ \&\ r_1\unif r_2\in \mathcal{U}\}$, $\depbnd{\Theta}
{\mathcal{U}} = \max \left\{\dep{t}\ \middle\vert\ t\in  \mathit{terms}(\mathcal{U})\right\}$ when $U\not = \emptyset$, and $\depbnd{\Theta}
{\mathcal{U}} = 0$ otherwise. 

\begin{lemma}
\label{lem:bounddepth}
Let $\uProb{\Theta}$ be a schematic unification problem such that $\Theta$ is primitive and $i\geq 0$ such that $\storeC{\Theta}{\useq{\Theta}{i}}\not = \bot$ . Then $\depbnd{\Theta}{\NstoreC{\Theta}{\useq{\Theta}{i}}} \leq  \depb$.
\end{lemma}
\begin{proof} When $i=0$,  $\depbnd{\Theta}{\NstoreC{\Theta}
{\useq{\Theta}{0}}} \leq  \depb$ follows from termination of $
\Theta$-unification (Theorem~\ref{thm:terminationTheta}) and 
Definition~\ref{def.lupunif}. We now assume, as our induction hypothesis, that $\depbnd{\Theta}{\NstoreC{\Theta}{\useq{\Theta}{i}}} \leq  \depb$ holds up to some $i$ and show that it also holds for $i+1$. Now consider $\depbnd{\Theta}{\NstoreC{\Theta}{\useq{\Theta}{i}}\stepsubs{i}{U}} \leq \depbnd{\Theta}{\useq{\Theta}{0}\stepsubs{0}{U}\stepsubs{i}{U}}$ derived from the induction hypothesis. From Theorem~\ref{thm:terminationTheta},  we derive $\depbnd{\Theta}{\storeC{\Theta}{\NstoreC{\Theta}{\useq{\Theta}{i}}\stepsubs{i}{U}}}\leq \depbnd{\Theta}{\NstoreC{\Theta}{\useq{\Theta}{i}}\stepsubs{i}{U}}$. 
Now, we normalize both sides of the inequality, resulting in $
\depbnd{\Theta}{\NstoreC{\Theta}{\storeC{\Theta}{\useq{\Theta}
{i}}\stepsubs{i}{U}}}\leq \depbnd{\Theta}{\useq{\Theta}{0}\stepsubs{0}{U}\stepsubs{i}{U}\uparrow_{\Theta}}$. Observe that normalization results in the following inequality $\depbnd{\Theta}{\NstoreC{\Theta}{\storeC{\Theta}
{\useq{\Theta}{i}}\stepsubs{i}{U}}}\leq \depbnd{\Theta}{\useq{\Theta}{0}\stepsubs{0}{U}}$. 
Using transitivity and $\depbnd{\Theta}{\NstoreC{\Theta}
{\useq{\Theta}{i+1}}}= \depbnd{\Theta}{\NstoreC{\Theta}{\storeC{\Theta}{\useq{\Theta}{i}}\stepsubs{i}{U}}}$, derived  from  Lemma~\ref{lem:delayedSubs}, we derive the following inequality,
$\depbnd{\Theta}{\NstoreC{\Theta}{\useq{\Theta}{i+1}}}\leq \depbnd{\Theta}{\useq{\Theta}{0}\stepsubs{0}{U}}$.
\end{proof}
\begin{lemma}
\label{lem:subbound}
Let $\uProb{\Theta}$ be a schematic unification problem such that $\Theta$ is primitive. Then for all $i\geq 0$ and for any $r\in 
\mathit{terms}(\NstoreC{\Theta}{\useq{\Theta}{i}})$, there exists $t\in \mathit{terms}(\useq{\Theta}{0}\stepsubs{0}{U})$, $t'\in \sub{t}$, and $k\geq 0$ such that $r = \shiftd{k}{t'}$.
\end{lemma}
\begin{proof}
   The rules of Table~\ref{tab:infrules} do not substitute into variables; thus, the term structure is preserved up to the shifting of indices.   
\end{proof}

\subsection{Bounding Variables Count}
In addition to bounding the term depth of terms occurring in $\useq{\Theta}{i}$, for large enough $i$, we can bound the number of \textit{fresh} variables occurring in $\useq{\Theta}{(i+1)}$, but not $i$. We conjecture that there is a bound on the number of variables that remain relevant to a given instance $\useq{\Theta}{i}$. If our conjectured bound is correct and we are considering a large enough $i$, then for all uniform schematic unification problems, there is a finite bound on the number of distinct variables in $\storeC{\Theta}{\useq{\Theta}{i}}$. To aid our presentation of the results  in this section, we introduce the following functions: $\maxbnd{\Theta}{U} = \max\left\{ \max \{i\mid i\in \getidx{t}\} \ \middle\vert\ \ t\in \mathit{terms}(U)\right\}$, and 
$\minbnd{\Theta}{U} = \min\left\{ \min \{i\mid i\in \getidx{t}\}\ \middle\vert\ \ t\in \mathit{terms}(U)\right\}$
where $U$ is a set of unification equations and $\Theta$ is a schematic unification.

\begin{lemma}\label{lem:OneMorePerClass}
Let $\uProb{\Theta}$ be a primitive schematic unification problem. Then for all $i\geq \maxbnd{\Theta}{\useq{\Theta}{0}\stepsubs{0}{U}}$, $\vert \vars{\useq{\Theta}{i+1}}\vert - \vert \vars{\useq{\Theta}{i}}\vert \leq \vert \getcls{\baseSet{\Theta}}\setminus\dom{\Theta}\vert$.
\end{lemma}
\begin{proof}
If $|\getcls{\useq{\Theta}{0}\stepsubs{0}{U}}|=0$, then $
\vars{\useq{\Theta}{i}}=\emptyset$ for all $i\geq 0$ and the bound 
trivially holds. We assume  for the induction hypothesis that for $i\geq 
\maxbnd{\Theta}{\useq{\Theta}{0}\stepsubs{0}{U}}$, $\vert 
\vars{\useq{\Theta}{i}}\vert -\vert \vars{\useq{\Theta}{i-1}}\vert \leq 
\vert \getcls{\baseSet{\Theta}}\setminus\dom{\Theta}\vert$ and show that 
this statement also holds for $i+1$. In other words, at most, one fresh 
variable per variable symbol. Observe that $\vars{\useq{\Theta}{i+1}} = 
(\vars{\useq{\Theta}{i}}\setminus\mathcal{P}) \cup \vars{\mathcal{Q}}$ 
where $\mathcal{P} = \{ x\mid x\in \Theta\ \& \ x\in \vars{\useq{\Theta}
{i}} \}$ and $\mathcal{Q} = \{\shiftd{j}{\base{\Theta}{X}}\mid X_j\in 
\mathcal{P} \}$. Thus, any variable occurring in $\useq{\Theta}{i+1}$ 
but not in $\useq{\Theta}{i} $ where introduced by $\mathcal{Q}$, i.e.  
at least one variable per symbol in $\getcls{\baseSet{\Theta}}$.

Now let $n_1,n_2,j,k\geq 0$ such that $j<k$, $\shiftd{n_1}{r_1},\shiftd{n_2}{r_2}\in \mathcal{Q}$, $X_j\in \vars{\shiftd{n_1}{r_1}}$, $X_k\in \vars{\shiftd{n_2}{r_2}}$, and
$X_j,X_k\not\in \vars{\useq{\Theta}{i}}$. Observe that $X\not\in \dom{\Theta}$ as we assume that $\Theta$ is primitive, i.e. for all $X\in \dom{\Theta}$, $\Rec{\Theta}{X}\subset\{0,1\}$. 

Observe that there exists $\minbnd{\Theta}{\useq{\Theta}{0}\stepsubs{0}{U}}\leq m_1< m_2\leq \maxbnd{\Theta}{\useq{\Theta}{0}\stepsubs{0}{U}}$ such that  $j=i+m_1$ and $k=i+m_2$, e.g. $m_2 = l_1+l_2$ where $l_1\in \{ j\mid x\in \vars{\ran{\stepsubs{0}{U}}}\ \& \ \getidx{x} =\{j\}\& \ \getcls{x} =\{X\}\}$ and $l_2 =\{j\mid x\in \dom{\stepsubs{0}{U}}\ \& \ $ $X\in \getcls{x\stepsubs{0}{U}} \ \& \ \getidx{x} =\{j\}\}$. Essentially, $l_2$ is an index of a variable symbol from $\dom{\Theta}$ occurring in  $\useq{\Theta}{0}$ and $l_1$ is an index of an occurrence of a variable with variable symbol $X$ in the substitution $\stepsubs{0}{U}$. Together, they define the initial shifting of variables with the symbol $X$.   This implies that  $X_{j} \in \useq{\Theta}{i-(m_2-m_1)}$, contradicting our assumption that 
$X_{j} \not \in \vars{\useq{\Theta}{i}}$ as $\vars{\useq{\Theta}{i-(m_2-m_1)}} \subset \vars{\useq{\Theta}{i}}$. \end{proof}

Lemma~\ref{lem:OneMorePerClass} implies that, for large enough $i$, the instance problems are well-behaved regarding variable introduction.

\begin{definition}[Future Relevant]\label{def:futureRelevant}
Let $\uProb{\Theta}$ be a schematic unification problem,  $i\geq 0$, $\storeC{\Theta}{\useq{\Theta}{i}}\not = \bot$. Then $x\in \vars{\storeC{\Theta}{\useq{\Theta}{i}}}$, such that $x\not\in \Theta$, is \emph{ future relevant} if (i) for some $y\in \dom{\stepsubs{i}{U}}$, $\getidx{x}\geq \getidx{y}$ and (ii)  $X\in\getcls{\baseSet{\Theta}(Y)}$ where $\getcls{y} = \{Y\}$. We denote the  future relevant set of $\storeC{\Theta}{\useq{\Theta}{i}}$ as $\FuR{i}$.
\end{definition}
\begin{example}
    Consider Example~\ref{ex:theta2}. Then $\FuR{2}=\{X_3\}$ in  $\storeC{\Theta}{\useq{\Theta}{2}}$.
\end{example}

\begin{definition}[Irrelevant Variables]\label{def:IrrVars}
Let $\uProb{\Theta}$ be a schematic unification\\ problem,  $i\geq 0$, $\storeC{\Theta}{\useq{\Theta}{i}}\not = \bot$. Then $\irrV{i} = \{x\mid x\in \vars{\irr{i}}\ \& \ x \not \in \Theta \ \& \ x\not \in \vars{\storeC{\Theta}{\useq{\Theta}{i}}}\}$.
\end{definition}
\begin{definition}[EQ]
Let $\uProb{\Theta}$ be a schematic unification problem, $i\geq 0$, $\storeC{\Theta}{\useq{\Theta}{i}}$ $\not = \bot$, and $x\in \vars{\storeC{\Theta}{\useq{\Theta}{i}}}$. Then $\mathit{EQ}_i^x = \{x\}\cup \{ y\mid x\unif y,y\unif x \in \storeC{\Theta}{\useq{\Theta}{i}}\}$.
\end{definition}

We use the equivalence classes to define a simplifying substitution, which is used to define a bound on the variables in $\storeC{\Theta}{\useq{\Theta}{i}}$ for large enough $i$.

\begin{definition}[Simplifying Substitution]\label{def:EQsubstitution}
Let $\uProb{\Theta}$ be a schematic unification problem, $i\geq 0$ s.t.  $\storeC{\Theta}{\useq{\Theta}{i}}\not = \bot$. Then 
$\mathit{EQ}_i\subseteq \vars{\storeC{\Theta}{\useq{\Theta}{i}}}$ s.t. for all $x \in \vars{\storeC{\Theta}{\useq{\Theta}{i}}}$,  $x\in \mathit{EQ}_i$ iff \emph{(i)} for all $y\in \mathit{EQ}_i$, $\mathit{EQ}_i^x\cap \mathit{EQ}_i^y =\emptyset$, \emph{(ii)} for all $y\in \mathit{EQ}_i^x$, $\getidx{x}\geq \getidx{y}$. Let $\subEQ{i} = \{ y\mapsto x\mid x\not = y \ \&\ x\in \mathit{EQ}_i \ \& \ y\in \mathit{EQ}_i^x  \}$.
\end{definition}

The simplifying substitution and the irrelevant variables capture all variables not relevant to $\storeC{\Theta}{\useq{\Theta}{i}}$. We define a ratio between the variables in the store of past instances and variables in the store of the current instance. Past refers instances prior to $\max\{\maxbnd{\Theta}{\useq{\Theta}{0}\stepsubs{0}{U}}$ $,\depbnd{\Theta}{\useq{\Theta}{0}\stepsubs{0}{U}}\}$,  denoted $\mathit{stab}(\uProb{\Theta})$.
\begin{definition}[Stability]\label{def:stability}
    Let $\uProb{\Theta}$ be a primitive schematic unification problem and $i= \mathit{stab}(\uProb{\Theta})$, $j\geq i$. We refer to $\useq{\Theta}{j}$ as \emph{stable} if $\storeC{\Theta}{\useq{\Theta}{j}}= \bot$ or $\mathit{stab}(\uProb{\Theta},j) \leq 1$ where  $$\mathit{stab}(\uProb{\Theta},j) = \min\left\{ \frac{ \vert \vars{\NstoreC{\Theta}{\useq{\Theta}{j}}\subEQ{j}}\vert   - \vert\irrV{j}\vert}{\vert \vars{\NstoreC{\Theta}{\useq{\Theta}{k}}\subEQ{k}}\vert   - \vert\irrV{k}\vert  } \ \middle\vert\ 0\leq k\leq i\right\}.$$
\end{definition}
 \emph{Stability} is based on the following observation: For every $t\in\baseSet{\Theta}$, if there exist positions $p,q\in\pos{t}$ and $i,k\geq 0$ such that $t\vert_p =X_i$ and $t\vert_q =X_{i+k}$, then $\shiftd{k}{t}\vert_p =X_{i+k}$, i.g. variables with large indices occur at the position of variables with lower indices in the shifted term. At instance $\mathit{stab}(\uProb{\Theta})$, all variables occuring in $\useq{\Theta}{0}\stepsubs{0}{U}$ which could shift positions have shifted positions. For large input terms, we have to wait until the terms have been totally decomposed, hence $\depbnd{\Theta}{\useq{\Theta}{0}\stepsubs{0}{U}}$.
\begin{definition}[$\infty$-Stability]\label{def:infstable}
    Let $\uProb{\Theta}$ be a primitive schematic unification\\ problem. Then $\uProb{\Theta}$ is \emph{$\infty$-stable} if for all $j\geq \mathit{stab}(\uProb{\Theta})$, $\storeC{\Theta}{\useq{\Theta}{j}}$ is stable. 
\end{definition}

\begin{theorem}
Let $\uProb{\Theta}$ be a uniform schematic unification problem and $\infty$-stable, $j\geq \mathit{stab}(\uProb{\Theta})$, and $i\geq 0$ is the index of the instance that minimizes $\mathit{stab}(\uProb{\Theta},j)$. Then if for all $k\geq 0$, $\storeC{\Theta}{\useq{\Theta}{k}}\not = \bot$ then  $\vert\vars{\NstoreC{\Theta}{\useq{\Theta}{j}}}\subEQ{j}\vert \leq \vert \vars{\NstoreC{\Theta}{\useq{\Theta}{i}}}\subEQ{i}\vert $.
\end{theorem}
\begin{proof}
    Observe that the index that minimizes $\mathit{stab}(\uProb{\Theta},j)$ is the same for all  $j\geq \mathit{stab}(\uProb{\Theta})$. By  Theorem~\ref{thm:primitiveEnough}, we can transform  $\uProb{\Theta}$ into a primitive schematic unification problem. When $j=i$, the statement trivially holds. Let us assume that the statement holds for $j$, for the induction hypothesis and show that it holds for $j+1$. If $\vert\vars{\NstoreC{\Theta}{\useq{\Theta}{(j+1)}}}\subEQ{(j+1)}\vert \leq \vert \vars{\NstoreC{\Theta}{\useq{\Theta}{j}}}\subEQ{j}\vert$, then the statement follows by transitivity with the induction hypothesis. If $\vert\vars{\NstoreC{\Theta}{\useq{\Theta}{j}}}\subEQ{j}\vert \leq \vert \vars{\NstoreC{\Theta}{\useq{\Theta}{(j+1)}}}\subEQ{(j+1)}\vert$, then we prove the statement as follows: By Lemma~\ref{lem:OneMorePerClass}, There are $\vert \getcls{\baseSet{\Theta}}\setminus\dom{\Theta}\vert$ variables occuring in $\useq{\Theta}{(j+1)}$ which did not occur in $\useq{\Theta}{j}$. Because $\useq{\Theta}{(j+1)}$ and $\useq{\Theta}{j}$ are stable we can derive the following:
   \scalebox{.95}{
 \begin{minipage}{1\textwidth}
   \begin{align}\label{eq:ineq}
        \vert\irrV{j+1}\vert \geq   (j-i)\cdot\vert \getcls{\baseSet{\Theta}}\setminus\dom{\Theta}\vert+ \vert\irrV{i}\vert\\ 
 \vert\irrV{j}\vert \geq   (j-i-1)\cdot\vert \getcls{\baseSet{\Theta}}\setminus\dom{\Theta}\vert+ \vert\irrV{i}\vert \\
 \vert\irrV{j+1}\vert \geq   \vert \getcls{\baseSet{\Theta}}\setminus\dom{\Theta}\vert+ \vert\irrV{j}\vert
\end{align}
\end{minipage}}

\vspace{.5em}
\noindent Substracting Line (2) from (1) and adding  $\vert\irrV{j}\vert$ to both sides results in Line (3), which entails $\vert\vars{\NstoreC{\Theta}{\useq{\Theta}{(j+1)}}}\subEQ{(j+1)}\vert \leq \vert \vars{\NstoreC{\Theta}{\useq{\Theta}{i}}}\subEQ{i}\vert$.
\end{proof}

\begin{conjecture}\label{conj:infStable}
All uniform schematic unification problems are  $\infty$-stable.
\end{conjecture}
If there exists a uniform schematic unification problem which is not  $\infty$-stable, for large $i$, $\storeC{\Theta}{\useq{\Theta}{i}}$ will contain chains of bindings $x^1\unif t_1,\dots, x^n\unif t_n$ where $t_1,\dots,t_n\not \in \Vars$, $x^{(j+1)}\in \vars{t_j}$ for $0 < j < n$, $x^1$ is relevant to computing $\storeC{\Theta}{\useq{\Theta}{(i+1)}}$, and some variables in $\vars{t_n}$ are not. Thus variables introduced by  $\useq{\Theta}{i-k}$  where $0\leq k\leq i$ may be relevant to $\useq{\Theta}{(i+1)}$. While there are unification problems producing binding chains, such as $\{f(f(f(f(a,X_0),Y_0),Z_0),W_0)\unif f(W_0,f(Z_0,f(Y_0,f(X_0,a))))\}$, the binding chain produced is inversely ordered to what we describe above. With shifting (Definition~\ref{def.subsshift}), the construction of such examples is highly non-trivial; we conjecture that none exist. 

\section{Unifiability of Uniform Schematic Unification Problems}
\label{sec:decide}
Our algorithm accepts input that is potentially not $\infty$-stable and throws an exception if stability is violated; $\mathit{uSchUnif}(\mathcal{V}_{\Xi})$ returns $\bot$ in such cases. The subprocedures of Algorithm~\ref{alg:procedure} are defined in Table~\ref{tab:mainProcedureSubs}.
\begin{algorithm}  
\SetAlgoLined\DontPrintSemicolon
\SetKwProg{myalg}{Algorithm}{}{}
 \myalg{$\mathit{uSchUnif}(\mathcal{V}_{\Xi})$\mbox{ \textbf{:}}}{ $ \uProb{\Theta}\leftarrow \mathit{makePrimitive}(\mathcal{V}_{\Xi})$\\
    $\config{\storeC{\Theta}{\useq{\Theta}{0}}}{\mathcal{U}_0}{\Theta} \leftarrow \mathit{thUnif}(\useq{\Theta}{0}, \Theta)$\tcp*{Throws Exception}
    $\stepsubs{0}{U} \leftarrow stepSubs(\storeC{\Theta}{\useq{\Theta}{0}},\Theta)$\\
        $s\ ,\ \irr{0} \leftarrow\mathit{stab}(\uProb{\Theta})\ ,\ \mathit{irr}(\storeC{\Theta}{\useq{\Theta}{0}},\mathcal{U}_0)$\\
    \For{$i\leftarrow 1$ \KwTo $\infty$ \KwBy 1}{
   $\config{\storeC{\Theta}{\useq{\Theta}{i}}}{\mathcal{U}_i}{\Theta} \leftarrow \mathit{thUnif}(\storeC{\Theta}{\useq{\Theta}{(i-1)}}\stepsubs{(i-1)}{U},\Theta)$         \tcp*{Throws Ex.}
   $\irr{i}\leftarrow \irr{(i-1)}\stepsubs{(i-1)}{U}\cup \mathit{irr}(\storeC{\Theta}{\useq{\Theta}{i}},\mathcal{U}_i)$\\
   $\mathit{cycle}(\irr{i},\Theta)$\tcp*{Throws Exception}
   $\subEQ{i},\FuR{i}\leftarrow \mathit{computeEqFr(\storeC{\Theta}{\useq{\Theta}{i}},\Theta)}$ \\
   $\stepsubs{i}{U} \leftarrow stepSubs(\storeC{\Theta}{\useq{\Theta}{i}},\Theta)$\\
  \If{$i\geq s$}{
  $ \mathit{stability}(i,s,\useq{\Theta}{i},\subEQ{i},\irr{i})$   \tcp*{Throws Exception}
  \For{$j\leftarrow 0$ \KwTo $i-1$ \KwBy 1}{
  $\mu \leftarrow \mathit{checkForMap}(\FuR{i},\NstoreC{\Theta}{\useq{\Theta}{i}}\subEQ{i},\FuR{j},\NstoreC{\Theta}{\useq{\Theta}{j}}\subEQ{j},\Theta)$\\
  \If{$\mu  \wedge \mathit{varDisjoint}(\NstoreC{\Theta}{\useq{\Theta}{i}}\subEQ{i},\NstoreC{\Theta}{\useq{\Theta}{j}}\subEQ{j},\Theta)$}{
    \Return\ $\mu,\ \storeC{\Theta}{\useq{\Theta}{i}},\ \subEQ{i},\ \storeC{\Theta}{\useq{\Theta}{j}},\ \subEQ{j},\ \irr{i}$ \\
}
  }
}
  }}
  \caption{Unification Algorithm for Uniform Schematic Unification}
\label{alg:procedure}
\vspace{-.04em}
\end{algorithm}
\begin{table}
    \centering

\begin{itemize}
    \item  $\mathit{makePrimitive}(\mathcal{V}_{\Xi})$-  The input $\mathcal{V}_{\Xi}$ is a uniform schematic unification problem. The procedure transforms $\Xi$ into a primitive schematic substitution $\Theta$ using the sustitution $\sigma = \sigma_1\sigma_2\dots \sigma_n$ where $\vert \dom{\Xi}\vert= n$ and $\sigma_1,\sigma_2,\dots, \sigma_n$ are the substitution computed using Theorem~\ref{thm:primitiveEnough}. Let $\sigma' =\{ x\mapsto t\mid x\sigma=t \ \&\ \getcls{x}\cap \dom{\Xi} = \emptyset\}$, i.g. $\sigma$ without domain bindings. Finally we compute substitution $\mathcal{U}=\mathcal{V}\sigma'$ and return $\uProb{\Theta}$.
    \item $\mathit{thUnif}(\mathcal{U},\Theta)$- The input is a set of unification equations $\mathcal{U}$ and a schematic substituion $\Theta$. The procedure applies $\Theta$-unification (Section~\ref{sec:findepunif}) to the initial configuration $\config{\emptyset}{\mathcal{U}}{\Theta}$ computing $\finCon{\Theta}{\mathcal{U}} = \config{\mathcal{S}}{\mathcal{\mathcal{U}}}{\Theta}$. If during the computation of $\finCon{\Theta}{\mathcal{U}}$ a clash is detected, the procedure throws an exception. Additionally, $\mathit{unif}(\mathcal{\mathcal{U}})$ is computed. If a cycle is detected, the procedure throws an exception. Otherwise $\finCon{\Theta}{\mathcal{U}}$ is returned.
    \item $\mathit{stepSub}(\mathcal{U},\Theta)$- The input is a set of unification equations $\mathcal{U}$ and a schematic substituion $\Theta$. The procedure computes the object presented in Definition~\ref{def:stepSub}.
    \item $\mathit{irr}(\mathcal{S},\mathcal{U})$- The input is the store $\mathcal{S}$ and active set  $\mathcal{U}$  produced by $\Theta$-unification (Section~\ref{sec:findepunif}). The procedure returns a set of bindings $\mathcal{I}$ (Definition~\ref{def:irrelevantSet}). Line 8 of Algorithm~\ref{alg:procedure} computes the subsequent irrelevant set using Lemma~\ref{lem:cumulativeI}.
   \item $\mathit{cycle}(\mathcal{U},\Theta)$- The input is a set of bindings $\mathcal{U}$ and a schematic substitution $\Theta$. We Check if $\mathcal{U}$ fails occurs check as follows: Let $I_{max}(\mathcal{U}',\Theta)= \max \{\getidx{x}\mid x\unif t \in \mathcal{U}'\}$, $R(\mathcal{U}',\Theta) = \{r\mid x\unif t \in \mathcal{U}'\wedge r\in\vars{t}\wedge r\in \Theta\} $, $R_{min}(\mathcal{U}',\Theta) = \min \{\getidx{r}\mid r\in R(\mathcal{U}',\Theta)\}$, $V_{gap}(\mathcal{U}',\Theta)= I_{max}(\mathcal{U}',\Theta)-R_{min}(\mathcal{U}',\Theta)+1$, $\mathit{chc}_{\Theta}(\mathcal{U}',0) = \mathcal{U}'$, and $\mathit{chc}_{\Theta}(\mathcal{U}',i+1) = \mathcal{U}^*(\bigcup_{r\in C} \Theta^r)$ where $\mathcal{U}^* = \mathit{chc}_{\Theta}(\mathcal{U}',i)$, and $C=R(\mathcal{U}^*,\Theta)$. An exception is thrown if $\mathit{unif}(\mathit{chc}_{\Theta}(\mathcal{U},V_{gap}(\mathcal{U},\Theta)) =\bot$ an exception is thrown. Essentially we are checking for \textit{occurs checks} in the Irrelevant set. 
    \item $\mathit{computeEqFr(\mathcal{S},\Theta)}$- The input is a set of bindings $\mathcal{S}$  and a schematic substitution $\Theta$. The procedure computes a substitution $\sigma$ presented in Definition~\ref{def:EQsubstitution} and the future relevant variables  (Definition~\ref{def:futureRelevant}).
    \item $ \mathit{stability}(i,s,\mathcal{U},\sigma,\mathcal{I})$- Takes as input $i,s\in \mathbb{N}$, variable mapping $\sigma$, and set of bindings $\mathcal{U}$, $\mathcal{I}$. Checks if $i\geq s$ and if this is the case, the procedure checks if $\mathcal{U}$, $\sigma$, and $\mathcal{I}$ are stable with respect to the problem with index minimizing $\mathit{stab}(\uProb{\Theta},i)$  (Definition~\ref{def:stability}). If the stability fails, an exception is thrown. 
    \item $\mathit{checkForMap}(\mathcal{F}_1,\mathcal{U}_1,\mathcal{F}_2,\mathcal{U}_2,\Theta)$ - Takes as input sets of variables $F_1$ and $\mathcal{F}_2$, set of bindings $\mathcal{U}_1$, $\mathcal{U}_2$, and a schematic substitution $\Theta$. The procedure searches for a mapping $\mu$ from $\vars{\mathcal{U}_1}$ to $\vars{\mathcal{U}_2}$ such that \emph{(i)} if $x\in\dom{\mu}$ and $x\in \mathcal{F}_1$ then $x\mu\in \mathcal{F}_2$,  \emph{(ii)} 
    $\getcls{x}=\getcls{x\mu}$,  \emph{(iii)} if $x\in \vars{\mathcal{U}_1}\cap \vars{\mathcal{U}_2}$, the $x=x\mu$ , and \emph{(iv)}  $\mathcal{U}_1\mu=\mathcal{U}_2$. If such a mapping $\mu$ exists,  return $\mu$, otherwise \textbf{None}.
    \item $\mathit{varDisjoint}(\mathcal{F}_1,\mathcal{U}_1,\mathcal{F}_2,\mathcal{U}_2,\Theta)$ - Takes as input sets of variables $F_1$ and $\mathcal{F}_2$, set of bindings $\mathcal{U}_1$, $\mathcal{U}_2$, and a schematic substitution $\Theta$.  The procedure checks if \emph{(i)} $\vars{\mathcal{U}_1}\cap \vars{\mathcal{U}_2}\cap (\mathcal{F}_1\cup \mathcal{F}_2)=\emptyset$ and \emph{(ii)} $\mathcal{F}_1\cap \mathcal{F}_2=\emptyset$ return \textbf{True} otherwise \textbf{False}. 
\end{itemize}
    \caption{Subprocedures of Algorithm~\ref{alg:procedure}.}
    \label{tab:mainProcedureSubs}
\end{table}
\begin{theorem}[Termination]\label{Thm:algTerm}
    Let $\uProb{\Theta}$ be a uniform schematic unification problem. Then  $\mathit{uSchUnif}(\uProb{\Theta})$ terminates.
\end{theorem}
\begin{proof}
    We assume that $\uProb{\Theta}$ is $\infty$-stable. If not, then for $i\geq \mathit{stab}(\uProb{\Theta})$, Line 13 of Algorithm~\ref{alg:procedure} will throw an exception, and the procedure will terminate. Concerning Lines 3 \& 7, we show that $\Theta$-unification terminates in Theorem~\ref{thm:terminationTheta}. By our assumption that  $\uProb{\Theta}$ is $\infty$-stable and Lemma~\ref{lem:bounddepth}, \ref{lem:subbound}, \& \ref{lem:OneMorePerClass}, we show that there exists is a finite set of sets of bindings $\mathbf{B}$ such that for all $i\geq \mathit{stab}(\uProb{\Theta})$, $\NstoreC{\Theta}{\useq{\Theta}{i}}\subEQ{i}\in \mathbf{B}$ modulo variable renaming (Definition~\ref{def:normalized}). Thus, by the infinitary pigeonhole principle, there exists $\mathit{stab}(\uProb{\Theta}) \leq i < j$ such that $\NstoreC{\Theta}{\useq{\Theta}{i}}\subEQ{i}= \NstoreC{\Theta}{\useq{\Theta}{j}}\subEQ{j}$, modulo variable renaming, and are disjoint (Line 16 of Algorithm~\ref{alg:procedure}). Thus, if no exception is thrown, Line 17 is always reached.
    \end{proof}
\begin{theorem}[Soundness]\label{thm:algsoundness}
    Let $\uProb{\Theta}$ be a uniform schematic unification problem. If $\mathit{uSchUnif}(\uProb{\Theta})\not = \bot$, then for all $i\geq 0$, $\useq{\Theta}{i}$ is unifiable.
\end{theorem}
\begin{proof}
Let us consider $\useq{\Theta}{0}$. Observe that Line 3 of Algorithm~\ref{alg:procedure} applies $\Theta$-unification to $\useq{\Theta}{0}$. Given that $\Theta$-unification is \textit{sound} (Theorem~\ref{Thm:soundnessTheta}) we can conclude that if $\mathit{thUnif}(\useq{\Theta}{0}) =\bot$ then  $\mathit{uSchUnif}(\uProb{\Theta}) = \bot$ (contrapositive). 

For the step case, we assume that for all $i\geq 0$, if $\mathit{uSchUnif}(\uProb{\Theta})\not = \bot$, then $\useq{\Theta}{i}$ is unifiable, and show that if $\mathit{uSchUnif}(\uProb{\Theta})\not = \bot$, then $\useq{\Theta}{(i+1)}$ is unifiable. Observe that, by the construction of $\useq{\Theta}{i}$ and $\useq{\Theta}{i+1}$, if $\mathit{thUnif}(\useq{\Theta}{i}) =\bot$ then $\mathit{thUnif}(\useq{\Theta}{i+1}) =\bot$ (Corollary~\ref{cor:alwaysfail}). In this case, the Theorem trivially follows. We now consider the case where $\mathit{thUnif}(\useq{\Theta}{i}) \not =\bot$ and  $\mathit{thUnif}(\useq{\Theta}{i+1}) =\bot$.  We assume that $\mathit{uSchUnif}(\uProb{\Theta}) \not = \bot$ and reach a contradiction. 

let us assume that $\mathit{uSchUnif}(\uProb{\Theta})$ terminates with a cycle between $\useq{\Theta}{m}$ and $\useq{\Theta}{n}$ for $0\leq m < n$. Now, we need to consider two cases:
\begin{itemize}
    \item $i+1\leq n:$  Observe that  $ \mathit{thUnif}(\useq{\Theta}{i},\Theta) = \config{\storeC{\Theta}{\useq{\Theta}{i}}}{\storeC{\Theta}{\useq{\Theta}{i}}\cup \irr{i}\cup N}{\Theta}$ where $N$ is as defined in the proof of Lemma~\ref{lem:delayedSubs} and can be ignored. Thus,  $\mathit{thUnif}(\useq{\Theta}{i+1}) =\bot$ due to  $\mathit{thUnif}(\storeC{\Theta}{\useq{\Theta}{i}}\stepsubs{i}{U},\Theta)=\bot$ or  $\mathit{cyclic}(\irr{(i+1)},\Theta)$ throwing an exception where $\irr{(i+1)}$ is built from the final configuration of $\mathit{thUnif}(\storeC{\Theta}{\useq{\Theta}{i}}\stepsubs{i}{U},\Theta)$ and $\irr{i}$. Either case contradicts $\mathit{uSchUnif}(\uProb{\Theta}) \not = \bot$. 
    
    \item $n< i+1:$ Note that $i = m+  k\cdot (n-m) +j$ for some $k\geq 1$ and $0\leq j < n-m$. By Algorithm~\ref{alg:procedure}, we know that there is a mapping from the variables of $\NstoreC{\Theta}{\useq{\Theta}{i}}\subEQ{i}$ to $\NstoreC{\Theta}{\useq{\Theta}{(m+j)}}\subEQ{(m+j)}$. Furthermore, Line 16  of Algorithm~\ref{alg:procedure} ensures that  $\NstoreC{\Theta}{\useq{\Theta}{i}}\subEQ{i}$ and $\NstoreC{\Theta}{\useq{\Theta}{(m+j)}}\subEQ{(m+j)}$ are variable disjoint and none of the variables occuring in $\NstoreC{\Theta}{\useq{\Theta}{(m+j)}}\subEQ{(m+j)}$ are future relevant to $\NstoreC{\Theta}{\useq{\Theta}{i}}\subEQ{i}$; this avoids the possibility of unchecked variable cycles. From this observation we can deduce that there is a mapping from $\mathit{irr}(\storeC{\Theta}{\useq{\Theta}{i}},\mathcal{U}_i)$ to $\mathit{irr}(\storeC{\Theta}{\useq{\Theta}{(m+j)}},\mathcal{U}_{(m+j)})$ as the same bindings will be irrelevant modulo variable renaming. 

  As the output of Algorithm~\ref{alg:procedure} defines a cycle, the above deduction implies that a mapping from the variables of $\NstoreC{\Theta}{\useq{\Theta}{(i+1)}}\subEQ{(i+1)}$ to the variables of $\NstoreC{\Theta}{\useq{\Theta}{(m+j+1)}}\subEQ{(m+j+1)}$ exists, thus equating the two unification problems. However, $\mathit{thUnif}(\storeC{\Theta}{\useq{\Theta}{(m+j+1)}}\stepsubs{(m+j+1)}{U},\Theta)\not=\bot$ as we know  
$\mathit{uSchUnif}(\uProb{\Theta})$ terminated with a cycle. Thus, if  $\mathit{thUnif}(\useq{\Theta}{i+1}) =\bot$, then  $\mathit{cyclic}(\irr{(i+1)},\Theta)$ must throw an exception. However, as with the $i^{th}$ case there is a mapping from $\mathit{irr}(\storeC{\Theta}{\useq{\Theta}{(i+1)}},\mathcal{U}_{(i+1)})$ to $\mathit{irr}(\storeC{\Theta}{\useq{\Theta}{(m+j+1)}},\mathcal{U}_{(m+j+1)})$ as the same bindings are introduced modulo variable renaming. Thus, if $\mathit{cyclic}(\irr{(i+1)},\Theta)$ throws an exception, then  $\mathit{cyclic}(\irr{(m+j+1)},\Theta)$ must also throw an exception contradicting our assumption that $\mathit{uSchUnif}(\uProb{\Theta}) \not = \bot$. Thus, we have derived that if $\useq{\Theta}{(i+1)}$ is not  unifiable and $\mathit{uSchUnif}(\uProb{\Theta}) = \bot$, the contrapositive of the induction hypothesis.
\end{itemize}
\end{proof}
\begin{theorem}[Completeness]
    Let $\uProb{\Theta}$ be an $\infty$-stable uniform schematic unification problem. If for all $i\geq 0$, $\useq{\Theta}{i}$ is unifiable, then $\mathit{uSchUnif}(\uProb{\Theta})$ $\not = \bot$.
\end{theorem}
\begin{proof}
Follows from the proofs of Theorem~\ref{Thm:algTerm}\ \& \ref{thm:algsoundness}.
\end{proof}
\begin{conjecture}[Completeness]
    Let $\uProb{\Theta}$ be a uniform schematic unification problem. If for all $i\geq 0$, $\useq{\Theta}{i}$ is unifiable, then $\mathit{uSchUnif}(\uProb{\Theta})\not = \bot$.
\end{conjecture}
If Conjecture~\ref{conj:infStable} is correct, then Line 13 of Algorithm~\ref{alg:procedure} never throws an exception. Thus, completeness is an immediate consequence. Concerning the complexity of Algorithm~\ref{alg:procedure}, observe that Line 15 is essentially a subsumption check (an \emph{NP-complete} problem). Thus,  Algorithm~\ref{alg:procedure} has an exponential running time. 
\section{Conclusion}
\label{sec:conclud}
In this paper, we present \textit{schematic unification problems}, a generalization of first-order unification allowing variables to be associated with a recursive definition. We define a well-behaved fragment and present a sound and terminating algorithm. The algorithm is complete under the additional restriction that the schematic unification problem is $\infty$-stable. We conjecture that all uniform schematic unification problems are $\infty$-stable. Future work includes addressing the open conjecture,  generalizing these results to \textit{simple} schematic unification problems and possibly problems with mutually recursive schematic substitutions. Currently, the algorithm produces the objects necessary for constructing a schematic unifier. We plan to address the construction of the unifier in the near future. Also, we plan to investigate how one could integrate such unification into a resolution calculus, likely a calculus similar to the one presented in~\cite{DBLP:journals/jar/CernaLL21}. This type of unification is essential to developing a cyclic resolution calculus for inexpressive theories of arithmetic~\cite{10.1007/11554554_8,DBLP:conf/fossacs/Simpson17}.

\bibliographystyle{splncs04}

\bibliography{references}

\appendix

\section{Discussion of Implementation}
An implementation of the algorithm presented in this paper may be found in the repository \href{https://github.com/Ermine516/Schematic-Unification}{github.com/Ermine516/Schematic-Unification}. The implementation contains a test suite of 35 problems  and an additional 3 problems that are non-uniform (Examples/simple). 

\section{Example Run}
Consider the Schematic unification problem 
$$ \uProb{\Theta} = \{f(X_4,L_0) \unif f(f(Y_3,Y_3),f(Y_0,f(Y_1,Y_0)))\}$$
where schematic substitution is as follows:
$$\Theta = \{ L_i \mapsto f(f(X_{i+1},f(Z_i,f(X_{i+1},f(X_i,f(Z_{i+1},X_i))))),L_{i+1})\mid i\geq 0 \}$$
and $\mathit{stab}(\uProb{\Theta}) = 8$.
In each subsection below, we present the objects produced by each step of  Algorithm~\ref{alg:procedure}.

\subsection{i=0}
\begin{itemize}
\item $\vert\vars{\useq{\Theta}{0}}\vert  - \vert\dom{\subEQ{0}}\vert - \vert\irrV{0}\vert  = 5$
    \item $\storeC{\Theta}{\useq{\Theta}{0}}$ is as follows
    $$\left\{\begin{array}{l}
      X_4\unif  f(Y_3,Y_3) \\
      L_0\unif f(Y_0,f(Y_1,Y_0))
    \end{array}\right\}$$
    \item $\irr{0}$ is empty
    \item $\FuR{0} =\{X_4\}$
    \item $\subEQ{0}$ is identity substitution

    \item $\stepsubs{0}{U}=\{ L_0 \mapsto f(f(X_{1},f(Z_0,f(X_{1},f(X_0,f(Z_{1},X_0))))),L_{1}) \}$

\end{itemize}

\subsection{i=1}
\begin{itemize}
\item  $\vert\vars{\useq{\Theta}{1}}\vert  - \vert\dom{\subEQ{1}}\vert - \vert\irrV{1}\vert  = 9$
    \item $\storeC{\Theta}{\useq{\Theta}{1}}$ is as follows
    $$\left\{\begin{array}{l}
      Y_0\unif f(X_{1},f(Z_0,f(X_{1},f(X_0,f(Z_{1},X_0))))) \\
      X_4\unif  f(Y_3,Y_3) \\
      L_1\unif f(Y_1,Y_0)
    \end{array}\right\}$$
    \item $\irr{1}$ is empty
    \item $\FuR{1} =\{X_4,Z_1,X_1\}$
    \item $\subEQ{1}$ is identity substitution
    \item $\stepsubs{1}{U}=\{ L_1 \mapsto f(f(X_{2},f(Z_1,f(X_{2},f(X_1,f(Z_{2},X_1))))),L_{2}) \}$
    \end{itemize}

\subsection{i=2}
\begin{itemize}
\item $\vert\vars{\useq{\Theta}{2}}\vert  - \vert\dom{\subEQ{2}}\vert - \vert\irrV{2}\vert  = 7$
    \item $\storeC{\Theta}{\useq{\Theta}{2}}$ is as follows
    $$\left\{\begin{array}{l}
      Y_0\unif f(X_{1},f(Z_0,f(X_{1},f(X_0,f(Z_{1},X_0))))) \\
      Y_0\unif L_2\\
      X_4\unif  f(Y_3,Y_3) \\
      L_2\unif f(X_{1},f(Z_0,f(X_{1},f(X_0,f(Z_{1},X_0)))))
    \end{array}\right\}$$
    \item $\irr{2}$ is as follows
        $$\left\{\begin{array}{l}
        Y_1\unif f(X_{2},f(Z_1,f(X_{2},f(X_1,f(Z_{2},X_1))))) 
    \end{array}\right\}$$
    \item $\FuR{2} =\{X_4\}$
    \item $\subEQ{2}$ is identity substitution
    \item $\stepsubs{2}{U}=\{ L_2 \mapsto f(f(X_{3},f(Z_2,f(X_{3},f(X_2,f(Z_{3},X_2))))),L_{3}) \}$
\end{itemize}

\subsection{i=3}
\begin{itemize}
\item $\vert\vars{\useq{\Theta}{3}}\vert  - \vert\dom{\subEQ{3}}\vert - \vert\irrV{3}\vert  = 10$
    \item $\storeC{\Theta}{\useq{\Theta}{3}}$ is as follows
    $$\left\{\begin{array}{l}
      X_1\unif f(X_{3},f(Z_2,f(X_{3},f(X_2,f(Z_{3},X_2))))) \\
      X_4\unif  f(Y_3,Y_3) \\
      L_3\unif f(Z_0,f(X_{1},f(X_0,f(Z_{1},X_0))))
    \end{array}\right\}$$
    \item $\irr{3}$ is as follows
        $$\irr{2}\stepsubs{2}{U}\cup\left\{\begin{array}{l}
        
        Y_0\unif f(f(X_{3},f(Z_2,f(X_{3},f(X_2,f(Z_{3},X_2))))),L_{3})
    \end{array}\right\}$$
    \item $\FuR{3} =\{X_4,X_3,Z_3\}$
    \item $\subEQ{3}$ is identity substitution
    \item $\stepsubs{3}{U}=\{ L_3 \mapsto f(f(X_{4},f(Z_3,f(X_{4},f(X_3,f(Z_{4},X_3))))),L_{4}) \}$
\end{itemize}
\subsection{i=4}
\begin{itemize}
\item $\vert\vars{\useq{\Theta}{4}}\vert  - \vert\dom{\subEQ{4}}\vert - \vert\irrV{4}\vert  = 10$
    \item $\storeC{\Theta}{\useq{\Theta}{4}}$ is as follows
    $$\left\{\begin{array}{l}
      X_1\unif f(X_{3},f(Z_2,f(X_{3},f(X_2,f(Z_{3},X_2))))) \\
      X_4\unif  f(Y_3,Y_3) \\
      L_4\unif f(X_{1},f(X_0,f(Z_{1},X_0)))
    \end{array}\right\}$$
    \item $\irr{4}$ is as follows
        $$\irr{3}\stepsubs{3}{U}\cup\left\{\begin{array}{l}
        Z_0 \unif f(X_4,f(Z_3,f(X_4,f(X_3,f(Z_4,X_3)))))
    \end{array}\right\}$$
    \item $\FuR{4} =\{X_4\}$
    \item $\subEQ{4}$ is identity substitution
    \item $\stepsubs{4}{U}=\{ L_4 \mapsto f(f(X_{5},f(Z_4,f(X_{5},f(X_4,f(Z_{5},X_4))))),L_{5}) \}$
\end{itemize}
\subsection{i=5}
\begin{itemize}
\item $\vert\vars{\useq{\Theta}{5}}\vert  - \vert\dom{\subEQ{5}}\vert - \vert\irrV{5}\vert  = 5$
    \item $\storeC{\Theta}{\useq{\Theta}{5}}$ is as follows
    $$\left\{\begin{array}{l}  
      L_5 \unif f(X_0,f(Z_{1},X_0))\\
      Z_3 \unif Z_5\\
      Z_5 \unif Z_3\\
      X_5 \unif X_3\\
      X_3 \unif X_5\\
    \end{array}\right\}$$
    \item $\irr{5}$ is as follows
        $$\irr{4}\stepsubs{4}{U}\cup\left\{\begin{array}{l}
          X_1 \unif f(X_{3},f(Z_2,f(X_{3},f(X_2,f(Z_{3},X_2))))) \\
          X_1 \unif f(X_5,f(Z_4,f(X_5,f(X_4,f(Z_5,X_4)))))\\
          X_4 \unif  f(Y_3,Y_3) \\
          X_2 \unif  f(Y_3,Y_3) \\ 
          X_4 \unif X_2\\
          X_2 \unif X_4\\
    \end{array}\right\}$$
    \item $\FuR{5} =\{X_5,Z_5\}$
    \item $\subEQ{5} =\{Z_3\mapsto Z_5,X_3\mapsto X_5\}$
    \item $\stepsubs{5}{U}=\{ L_5 \mapsto f(f(X_{6},f(Z_5,f(X_{6},f(X_5,f(Z_{6},X_5))))),L_{6}) \}$
\end{itemize}

\subsection{i=6}
\begin{itemize}
\item $\vert\vars{\useq{\Theta}{6}}\vert  - \vert\dom{\subEQ{6}}\vert - \vert\irrV{6}\vert  = 8$
    \item $\storeC{\Theta}{\useq{\Theta}{6}}$ is as follows
    $$\left\{\begin{array}{l}  
      L_6 \unif f(Z_{1},X_0)\\
      X_0\unif f(X_{6},f(Z_5,f(X_{6},f(X_5,f(Z_{6},X_5))))\\
      Z_3 \unif Z_5\\
      Z_5 \unif Z_3\\
      X_5 \unif X_3\\
      X_3 \unif X_5\\
    \end{array}\right\}$$
    \item $\irr{6}$ is as follows
        $$\irr{5}\stepsubs{5}{U}$$
    \item $\FuR{6} =\{X_6,Z_6\}$
    \item $\subEQ{6} =\{Z_3\mapsto Z_5,X_3\mapsto X_5\}$
    \item $\stepsubs{6}{U}=\{ L_6 \mapsto f(f(X_{7},f(Z_6,f(X_{7},f(X_6,f(Z_{7},X_6))))),L_{7}) \}$
\end{itemize}
\subsection{i=7}
\begin{itemize}
\item $\vert\vars{\useq{\Theta}{7}}\vert  - \vert\dom{\subEQ{7}}\vert - \vert\irrV{7}\vert  = 7$
    \item $\storeC{\Theta}{\useq{\Theta}{7}}$ is as follows
    $$\left\{\begin{array}{l}  
      L_7 \unif f(X_{6},f(Z_5,f(X_{6},f(X_5,f(Z_{6},X_5)))) \\
      X_0\unif L_7\\
      Z_3 \unif Z_5\\
      Z_5 \unif Z_3\\
      X_5 \unif X_3\\
      X_3 \unif X_5\\
    \end{array}\right\}$$
    \item $\irr{7}$ is as follows
        $$\irr{6}\stepsubs{6}{U}\cup\left\{\begin{array}{l}
          Z_1 \unif f(X_7,f(Z_6,f(X_7,f(X_6,f(Z_7,X_6)))))\\
    \end{array}\right\}$$
    \item $\FuR{7} =\{\}$
    \item $\subEQ{7} =\{Z_3\mapsto Z_5,X_3\mapsto X_5\}$
    \item $\stepsubs{7}{U}=\{ L_7 \mapsto f(f(X_{8},f(Z_7,f(X_{8},f(X_7,f(Z_{8},X_7))))),L_{8}) \}$
\end{itemize}

\subsection{i=8}
\begin{itemize}
\item $\vert\vars{\useq{\Theta}{8}}\vert  - \vert\dom{\subEQ{8}}\vert - \vert\irrV{8}\vert  = 10$
    \item $\storeC{\Theta}{\useq{\Theta}{8}}$ is as follows
    $$\left\{\begin{array}{l}  
      L_8 \unif f(Z_5,f(X_{6},f(X_5,f(Z_{6},X_5)))) \\
      X_6 \unif f(X_{8},f(Z_7,f(X_{8},f(X_7,f(Z_{8},X_7)))))\\
      Z_3 \unif Z_5\\
      Z_5 \unif Z_3\\
      X_5 \unif X_3\\
      X_3 \unif X_5\\
    \end{array}\right\}$$
    \item $\irr{8}$ is as follows
        $$\irr{7}\stepsubs{7}{U}\cup\left\{\begin{array}{l}
         X_0\unif f(X_{6},f(Z_5,f(X_{6},f(X_5,f(Z_{6},X_5))))\\
          X_0\unif f(f(X_{8},f(Z_7,f(X_{8},f(X_7,f(Z_{8},X_7))))),L_{8})
    \end{array}\right\}$$
    \item $\FuR{9} =\{X_8,Z_8\}$
    \item $\subEQ{9} =\{Z_3\mapsto Z_5,X_3\mapsto X_5\}$
    \item $\stepsubs{9}{U}=\{ L_8 \mapsto f(f(X_{9},f(Z_8,f(X_{9},f(X_8,f(Z_{9},X_8))))),L_{9}) \}$
\end{itemize}
\subsection{i=9}
\begin{itemize}
    \item $\vert\vars{\useq{\Theta}{9}}\vert  - \vert\dom{\subEQ{9}}\vert - \vert\irrV{9}\vert  = 5$
    \item $\mathit{stab}(\uProb{\Theta},9)= .5$
    \item $\storeC{\Theta}{\useq{\Theta}{9}}$ is as follows
    $$\left\{\begin{array}{l}  
      L_9 \unif f(X_{6},f(X_5,f(Z_{6},X_5))) \\
      X_6 \unif f(X_{8},f(Z_7,f(X_{8},f(X_7,f(Z_{8},X_7)))))\\
      X_5 \unif X_3\\
      X_3 \unif X_5\\
    \end{array}\right\}$$
    \item $\irr{9}$ is as follows
        $$\irr{8}\stepsubs{8}{U}\cup\left\{\begin{array}{l}
          Z_3 \unif Z_5\\
          Z_5 \unif Z_3\\
          Z_3 \unif f(X_9,f(Z_8,f(X_9,f(X_8,f(Z_9,X_8))))) \\
	       Z_5\unif f(X_9,f(Z_8,f(X_9,f(X_8,f(Z_9,X_8)))))
    \end{array}\right\}$$
    \item $\FuR{9} =\{\}$
    \item $\subEQ{9} =\{X_3\mapsto X_5\}$
    \item $\stepsubs{9}{U}=\{ L_{9} \mapsto f(f(X_{10},f(Z_9,f(X_{10},f(X_9,f(Z_{10},X_9))))),L_{10}) \}$
\end{itemize}
\subsection{i=10}
\begin{itemize}
\item $\vert\vars{\useq{\Theta}{10}}\vert  - \vert\dom{\subEQ{10}}\vert - \vert\irrV{10}\vert  = 5$
    \item $\mathit{stab}(\uProb{\Theta},11)= .5$
    \item $\storeC{\Theta}{\useq{\Theta}{10}}$ is as follows
    $$\left\{\begin{array}{l}  
      L_{10} \unif f(X_5,f(Z_{6},X_5)) \\
      X_5 \unif X_3\\
      X_3 \unif X_5\\
      X_{8} \unif X_{10}\\
	Z_{10} \unif Z_{8}\\
	X_{10} \unif X_{8}\\
	Z_{8} \unif Z_{10}\\
    \end{array}\right\}$$
    \item $\irr{10}$ is as follows
        $$\irr{9}\stepsubs{9}{U}\cup\left\{\begin{array}{l}
        Z_{9}\unif Z_{7}\\
	X_{6}\unif f(X_{10},f(Z_{9},f(X_{10},f(X_{9},f(Z_{10},X_{9})))))\\
	X_{9}\unif X_{7}\\
	X_{7}\unif X_{9}\\
	Z_{7}\unif Z_{9}\\
	X_{6}\unif f(X_{8},f(Z_{7},f(X_{8},f(X_{7},f(Z_{8},X_{7})))))\\
    \end{array}\right\}$$
    \item $\FuR{10} =\{X_{10},Z_{10}\}$
    \item $\subEQ{10} =\{X_3\mapsto X_5,X_8\mapsto X_{10},Z_8\mapsto Z_{10}\}$
    \item $\stepsubs{10}{U}=\{ L_{10} \mapsto f(f(X_{11},f(Z_{10},f(X_{11},f(X_{10},f(Z_{11},X_{10}))))),L_{11}) \}$
\end{itemize}
\subsection{i=11}
\begin{itemize}
\item $\vert\vars{\useq{\Theta}{11}}\vert  - \vert\dom{\subEQ{11}}\vert - \vert\irrV{11}\vert  = 6$
    \item $\mathit{stab}(\uProb{\Theta},11)= .6$
    \item $\storeC{\Theta}{\useq{\Theta}{11}}$ is as follows
    $$\left\{\begin{array}{l}  
      L_{11} \unif f(Z_{6},X_5) \\
      X_3 \unif f(X_{11},f(Z_{10},f(X_{11},f(X_{10},f(Z_{11},X_{10})))))\\
      X_5 \unif f(X_{11},f(Z_{10},f(X_{11},f(X_{10},f(Z_{11},X_{10})))))\\
      X_5 \unif X_3\\
      X_3 \unif X_5\\
      X_{8} \unif X_{10}\\
	Z_{10} \unif Z_{8}\\
	X_{10} \unif X_{8}\\
	Z_{8} \unif Z_{10}\\
    \end{array}\right\}$$
    \item $\irr{11}$ is as follows
       $$\irr{10}\stepsubs{10}{U}$$
    \item $\FuR{11} =\{X_{11},Z_{11}\}$
    \item $\subEQ{11} =\{X_3\mapsto X_5,X_8\mapsto X_{10},Z_8\mapsto Z_{10}\}$
    \item $\stepsubs{11}{U}=\{ L_{11} \mapsto f(f(X_{12},f(Z_{11},f(X_{12},f(X_{11},f(Z_{12},X_{11}))))),L_{12}) \}$
\end{itemize}

\subsection{i=12}
\begin{itemize}
\item $\vert\vars{\useq{\Theta}{12}}\vert  - \vert\dom{\subEQ{12}}\vert - \vert\irrV{12}\vert  = 6$
    \item $\mathit{stab}(\uProb{\Theta},12)= .6$
    \item $\storeC{\Theta}{\useq{\Theta}{12}}$ is as follows
    $$\left\{\begin{array}{l}  
      X_5 \unif L_{12}   \\
       X_3 \unif L_{12}   \\
       L_{12}\unif f(X_{11},f(Z_{10},f(X_{11},f(X_{10},f(Z_{11},X_{10})))))\\
      X_3 \unif f(X_{11},f(Z_{10},f(X_{11},f(X_{10},f(Z_{11},X_{10})))))\\
      X_5 \unif f(X_{11},f(Z_{10},f(X_{11},f(X_{10},f(Z_{11},X_{10})))))\\
      X_5 \unif X_3\\
      X_3 \unif X_5\\
      X_{8} \unif X_{10}\\
	Z_{10} \unif Z_{8}\\
	X_{10} \unif X_{8}\\
	Z_{8} \unif Z_{10}\\
    \end{array}\right\}$$
    \item $\irr{12}$ is as follows
        $$\irr{11}\stepsubs{11}{U}\cup\left\{\begin{array}{l}
 	Z_{6} \unif f(X_{12},f(Z_{11},f(X_{12},f(X_{11},f(Z_{12},X_{11})))))\\
    \end{array}\right\}$$
    \item $\FuR{12} =\{\}$
    \item $\subEQ{12} =\{X_3\mapsto X_5,X_8\mapsto X_{10},Z_8\mapsto Z_{10}\}$
    \item $\stepsubs{12}{U}=\{ L_{12} \mapsto f(f(X_{13},f(Z_{12},f(X_{13},f(X_{12},f(Z_{13},X_{12}))))),L_{13}) \}$
        
\end{itemize}
\subsection{i=13}
\begin{itemize}
\item $\vert\vars{\useq{\Theta}{13}}\vert  - \vert\dom{\subEQ{13}}\vert - \vert\irrV{13}\vert  = 8$
    \item $\mathit{stab}(\uProb{\Theta},13)= .8$
    \item $\storeC{\Theta}{\useq{\Theta}{13}}$ is as follows
    $$\left\{\begin{array}{l}  
       L_{13}\unif f(Z_{10},f(X_{11},f(X_{10},f(Z_{11},X_{10}))))\\
      X_{8} \unif X_{10}\\
	Z_{10} \unif Z_{8}\\
	X_{10} \unif X_{8}\\
	Z_{8} \unif Z_{10}\\
 X_{11} \unif f(X_{13},f(Z_{12},f(X_{13},f(X_{12},f(Z_{13},X_{12})))))\\
    \end{array}\right\}$$
    \item $\irr{13}$ is as follows
        $$\irr{12}\stepsubs{12}{U}\cup\left\{\begin{array}{l}
    X_3 \unif f(X_{11},f(Z_{10},f(X_{11},f(X_{10},f(Z_{11},X_{10})))))\\
     X_5 \unif f(X_{11},f(Z_{10},f(X_{11},f(X_{10},f(Z_{11},X_{10})))))\\
      X_5 \unif X_3\\
      X_3 \unif X_5\\
    X_5 \unif f(f(X_{12},f(Z_{11},f(X_{12},f(X_{11},f(Z_{12},X_{11}))))),L_{12})\\
      X_3 \unif f(f(X_{12},f(Z_{11},f(X_{12},f(X_{11},f(Z_{12},X_{11}))))),L_{12})\\
    \end{array}\right\}$$
    \item $\FuR{13} =\{X_{13},Z_{13}\}$
    \item $\subEQ{13} =\{X_8\mapsto X_{10},Z_8\mapsto Z_{10}\}$
    \item $\stepsubs{13}{U}=\{ L_{13} \mapsto f(f(X_{14},f(Z_{13},f(X_{14},f(X_{13},f(Z_{14},X_{13}))))),L_{14}) \}$
\end{itemize}
\subsection{i=14}
\begin{itemize}
\item $\vert\vars{\useq{\Theta}{14}}\vert  - \vert\dom{\subEQ{14}}\vert - \vert\irrV{14}\vert  = 7$
    \item $\mathit{stab}(\uProb{\Theta},14)= .7$
    \item $\storeC{\Theta}{\useq{\Theta}{14}}$ is as follows
    $$\left\{\begin{array}{l}  
       L_{14}\unif f(X_{11},f(X_{10},f(Z_{11},X_{10})))\\
      X_{8} \unif X_{10}\\
	X_{10} \unif X_{8}\\
 X_{11} \unif f(X_{13},f(Z_{12},f(X_{13},f(X_{12},f(Z_{13},X_{12})))))\\
    \end{array}\right\}$$
    \item $\irr{14}$ is as follows
        $$\irr{13}\stepsubs{13}{U}\cup\left\{\begin{array}{l}
      Z_{8} \unif Z_{10}\\
 	Z_{10} \unif Z_{8}\\
       Z_{8} \unif f(X_{14},f(Z_{13},f(X_{14},f(X_{13},f(Z_{14},X_{13})))))\\
 	Z_{10} \unif f(X_{14},f(Z_{13},f(X_{14},f(X_{13},f(Z_{14},X_{13})))))\\
    \end{array}\right\}$$
    \item $\FuR{14} =\{\}$
    \item $\subEQ{14} =\{X_8\mapsto X_{10}\}$
    \item $\stepsubs{14}{U}=\{ L_{14} \mapsto f(f(X_{15},f(Z_{14},f(X_{15},f(X_{14},f(Z_{15},X_{14}))))),L_{15}) \}$
\end{itemize}
\subsection{i=15}
\begin{itemize}
    \item $\vert\vars{\useq{\Theta}{15}}\vert  - \vert\dom{\subEQ{15}}\vert - \vert\irrV{15}\vert  = 5$
    \item $\mathit{stab}(\uProb{\Theta},15)= .5$
    \item $\storeC{\Theta}{\useq{\Theta}{15}}$ is as follows
    $$\left\{\begin{array}{l}  
       L_{15}\unif f(X_{10},f(Z_{11},X_{10}))\\
      X_{8} \unif X_{10}\\
	X_{10} \unif X_{8}\\
X_{13}\unif X_{15}\\
	X_{15} \unif X_{13}\\
 Z_{15} \unif Z_{13}\\
	Z_{13} \unif Z_{15}\\
    \end{array}\right\}$$
    \item $\irr{15}$ is as follows
        $$\irr{14}\stepsubs{14}{U}\cup\left\{\begin{array}{l}
   X_{11} \unif f(X_{13},f(Z_{12},f(X_{13},f(X_{12},f(Z_{13},X_{12})))))\\
   Z_{14} \unif Z_{12}\\
	X_{11} \unif f(X_{15},f(Z_{14},f(X_{15},f(X_{14},f(Z_{15},X_{14})))))\\
	Z_{12} \unif Z_{14}\\
	X_{12} \unif X_{14}\\
	X_{14} \unif X_{12}\\
    \end{array}\right\}$$
    \item $\FuR{15} =\{X_{15},Z_{15}\}$
    \item $\subEQ{15} =\{X_8\mapsto X_{10},X_{13}\mapsto X_{15},Z_{13}\mapsto Z_{15}\}$
    \item $\stepsubs{15}{U}=\{ L_{15} \mapsto f(f(X_{16},f(Z_{15},f(X_{16},f(X_{15},f(Z_{16},X_{15}))))),L_{16}) \}$
    \item $\NstoreC{\Theta}{\useq{\Theta}{15}}$ is as follows
    $$\left\{\begin{array}{l}  
       L_{15}\unif f(X_{10},f(Z_{11},X_{10}))\\
    \end{array}\right\}$$
     \item $\NstoreC{\Theta}{\useq{\Theta}{5}}$ is as follows
    $$\left\{\begin{array}{l}  
       L_{5}\unif f(X_{0},f(Z_{1},X_{0}))\\
    \end{array}\right\}$$
    \item $\mu = \{X_{10}\mapsto X_{5},Z_{11}\mapsto Z_{1},L_{15}\mapsto L_{5}\}$
\end{itemize}
At this point the algorithm terminates. 

\end{document}